\documentclass[11pt, letterpaper, onecolumn, peerreview]{IEEEtran}

\usepackage{graphics}
\usepackage{amsmath}
\usepackage{amssymb}
\usepackage{graphicx}
\usepackage{cite}

\newtheorem{theorem}{\bf Theorem}
\newtheorem{lemma}{\bf Lemma}

\newtheorem{definition}{\bf Definition}

\newcommand{\m}[1]{\mathbf{#1}_1^m}
\newcommand{\lo}[1]{\log_2\left(#1\right)}
\newcommand{\lon}[1]{\ln\left(#1\right)}
\newcommand{\mk}[1]{\mathbf{#1}_1^k}
\newcommand{\mn}[1]{\mathbf{#1}_1^n}

\newcommand{\nbd}[1]{\mathbf{#1}_{\text{nbd}, i}^n}
\newcommand{\typical}{ \mathcal{T}_{\epsilon, G}}
\newcommand{\typicalc}{ \mathcal{T}_{\epsilon, G}^c}
\newcommand{\pe}{{\langle P_e\rangle}}
\newcommand{\pei}{{\langle P_{e,i}\rangle}}

\title{The price of certainty: ``waterslide curves'' and the gap to capacity}
\author{Anant Sahai and Pulkit Grover \\ Wireless Foundations, Department of EECS
\\University of California at Berkeley, CA-94720, USA\\\{sahai,
pulkit\}@eecs.berkeley.edu}

\begin{document}
\maketitle

\begin{abstract}
The classical problem of reliable point-to-point digital communication
is to achieve a low probability of error while keeping the rate high
and the {\bf total power consumption} small. Traditional
information-theoretic analysis uses explicit models for the
communication channel to study the power spent in 
transmission. The resulting bounds are expressed using `waterfall'
curves that convey the revolutionary idea that unboundedly low
probabilities of bit-error are attainable using only finite transmit
power. However, practitioners have long observed that the decoder
complexity, and hence the total power consumption, goes up when
attempting to use sophisticated codes that operate close to the
waterfall curve.

This paper gives an explicit model for power consumption at an
idealized decoder that allows for extreme parallelism in
implementation. The decoder architecture is in the spirit of message
passing and iterative decoding for sparse-graph codes, but is further
idealized in that it allows for more computational power than is
currently known to be implementable. Generalized sphere-packing
arguments are used to derive lower bounds on the decoding power needed
for any possible code given only the gap from the Shannon limit and
the desired probability of error. As the gap goes to zero, the {\em
  energy per bit} spent in decoding is shown to go to infinity. This
suggests that to optimize total power, the transmitter should operate
at a power that is strictly above the minimum demanded by the Shannon
capacity.

The lower bound is plotted to show an unavoidable tradeoff between the
average bit-error probability and the total power used in transmission
and decoding. In the spirit of conventional waterfall curves, we call
these `waterslide' curves. The bound is shown to be order optimal by
showing the existence of codes that can achieve similarly shaped
waterslide curves under the proposed idealized model of decoding.
\end{abstract}

\IEEEpeerreviewmaketitle

\thanks{{\em Note: A preliminary version of this work with weaker
    bounds was submitted to ITW 2008 in Porto \cite{ITWpaper08}.} }

\section{Introduction}
As digital circuit technology advances and we pass into the era of
billion-transistor chips, it is clear that the fundamental limit on
practical codes is not any nebulous sense of ``complexity'' but the
concrete issue of power consumption.  At the same time, the proposed
applications for error-correcting codes continue to shrink in the
distances involved. Whereas earlier ``deep space communication''
helped stimulate the development of information and coding theory
\cite{Massey92DeepSpace, McElieceShannonLecture}, there is now an
increasing interest in communication over much shorter distances
ranging from a few meters \cite{HowardSchlegel} to even a few
millimeters in the case of inter-chip and on-chip communication
\cite{CheeBusPaper}.



The implications of power-consumption beyond transmit power have begun
to be studied by the community. The common thread in \cite{pagrawal,
  goldsmithbahai,goldsmithwicker, massaad, Vasudevan} is that the
power consumed in processing the signals can be a substantial fraction
of the total power. In~\cite{mobydick}, it is observed that within
communication networks, it is worth developing cross-layer schemes to
reduce the time that devices spend being active.  In~\cite{massaad},
an information-theoretic formulation is considered.  When the
transmitter is in the `on' state, its circuit is modeled as consuming
some fixed power in addition to the power radiated in the transmission
itself. Therefore, it makes sense to shorten the overall duration of a
packet transmission and to satisfy an average transmit-power
constraint by bursty signalling that does not use all available
degrees of freedom.  In~\cite{goldsmithbahai}, the authors take into
account a peak-power constraint as well, as they study the optimal
constellation size for uncoded transmission. A large constellation
requires a smaller `on' time, and hence less circuit power. However, a
larger constellation requires higher power to maintain the same
spacing of constellation points. An optimal constellation has to
balance between the two, but overall this argues for the use of higher
rates. However, none of these really tackle the role of the decoding
complexity itself.

In~\cite{manishAllerton07}, the authors take a more receiver-centric
view and focus on how to limit the power spent in sampling the signal
at the receiver. They point out that empirically for ultrawideband
systems aiming for moderate probabilities of error, this sampling cost
can be larger than the decoding cost! They introduce the ingenious
idea of adaptively puncturing the code at the receiver rather than at
the transmitter. They implicitly argue for the use of longer codes
whose rates are further from the Shannon capacity so that the decoder
has the flexibility to adaptively puncture as needed and thereby save
on total power consumption.

In~\cite{HowardSchlegel}, the authors study the impact of decoding
complexity using the metric of coding gain. They take an empirical
point of view using power-consumption numbers for certain decoder
implementations at moderately low probabilities of error. They observe
that it is often better to use no coding at all if the communication
range is low enough.

In this paper, we take an asymptotic approach to see if considering
decoding power has any fundamental implications as the average
probability of bit error tends to zero. In
Section~\ref{sec:capacityachieving}, we give an asymptotic formulation
of what it should mean to approach capacity when we must consider the
power spent in decoding in addition to that spent in transmission. We
next consider whether classical approaches to encoding/decoding such
as dense linear block codes and convolutional codes can satisfy our
stricter standard of approaching capacity and argue that they cannot.
Section~\ref{sec:sysmod} then focuses our attention on iterative
decoding by message passing and defines the system model for the rest
of the paper.

Section~\ref{sec:lowbound} derives general lower bounds to the
complexity of iterative decoders for BSC and AWGN channels in terms of
the number of iterations required to achieve a desired probability of
error at a given transmit power. These bounds can be considered
iterative-decoding counterparts to the classical sphere-packing bounds
(see e.g.~\cite{Gallager, csiszarkorner}) and are derived by
generalizing the delay-oriented arguments of \cite{PinskerNoFeedback,
  OurUpperboundPaper} to the decoding neighborhoods in iterative
decoding.  These bounds are then used to show that it is in principle
possible for iterative decoders to be a part of a weakly
capacity-achieving communication system. However, the power spent by
our model of an iterative decoder must go to infinity as the
probability of error tends to zero and so this style of decoding rules
out a strong sense of capacity-achieving communication systems.

We discuss related work in the sparse-graph-code context in
Section~\ref{sec:gapandrelated} and make precise the notion of gap to
capacity before evaluating our lower-bounds on the number of
iterations as the gap to capacity closes. We conclude in
Section~\ref{sec:conclusions} with some speculation and point out some
interesting questions for future investigation.

\section{Certainty-achieving codes} \label{sec:capacityachieving}
Consider a classical point-to-point AWGN channel with no fading. For
uncoded transmission with BPSK signaling, the probability of bit-error
is an exponentially decreasing function of the transmitted energy per
symbol. To approach certainty (make the probability of bit-error very
small), the transmitted energy per symbol must go to infinity. If the
symbols each carry a small number of bits, then this implies that
the transmit {\em power} is also going to infinity since the number of
symbols per second is a nonzero constant determined by the desired
rate of $R$ bits per second.

Shannon's genius in \cite{ShannonOriginalPaper} was to recognize that
while there was no way to avoid having the transmitted {\em energy} go
to infinity and still approach certainty, this energy could be
amortized over many bits of information. This meant that the
transmitted {\em power} could be kept finite and certainty could be
approached by paying for it using end-to-end delay (see
\cite{OurUpperboundPaper} for a review) and whatever implementation
complexity is required for the encoding and decoding. For a given
channel and transmit power $P_T$, there is a maximum rate $C(P_T)$
that can be supported. Turned around, this classical result is
traditionally expressed by fixing the desired rate $R$ and looking at
the required transmit power. The resulting ``waterfall curves'' are
shown\footnote{Since the focus of this paper is on average bit error
    probability, these curves combine the results of
    \cite{ShannonLossy, ShannonOriginalPaper} and adjust the required
    capacity by a factor of the relevant rate-distortion function
    $1-h_b(\pe)$.} in Figure~\ref{fig:waterfall1}. These sharp curves are
distinguished from the more gradual ``waterslide curves'' of uncoded
transmission.

\begin{figure}[htb]
\begin{center}
\includegraphics[scale=0.7]{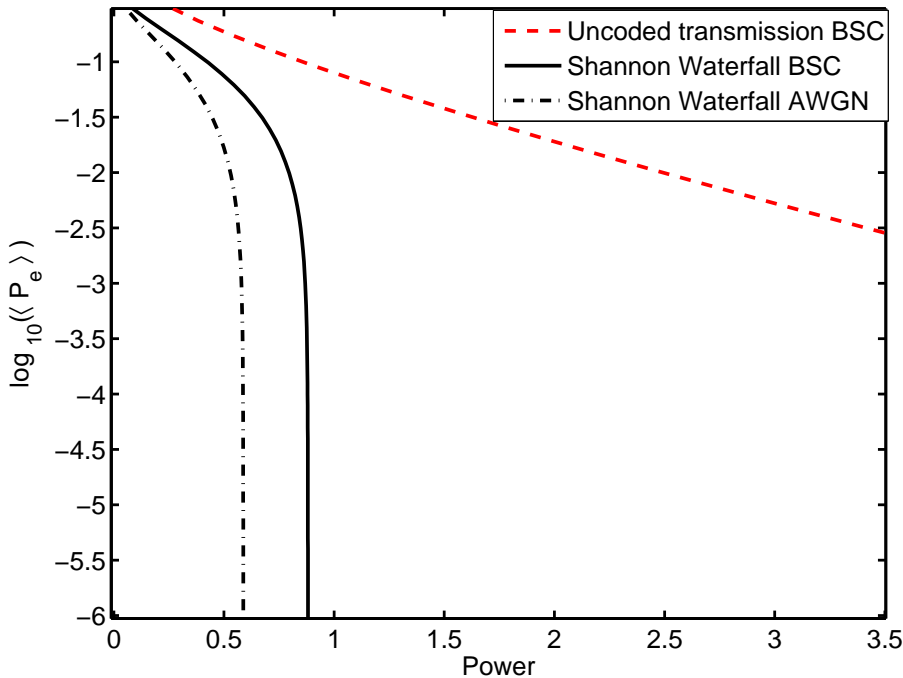}
\caption{The Shannon waterfalls: plots of $\log(\pe)$ vs
  required SNR (in dB) for a fixed rate-$1/3$ code
  transmitted using BPSK over an AWGN channel with hard decisions at
  the detector. A comparison is made with the rate-$1/3$ repetition
  code: uncoded transmission with the same bit repeated three
  times. Also shown is the waterfall curve for the average power
  constrained AWGN channel.}
\label{fig:waterfall1}
\end{center}
\end{figure}
Traditionally, a family of codes was considered capacity achieving if
it could support arbitrarily low probabilities of error at transmit
powers arbitrarily close to that predicted by capacity. The complexity
of the encoding and decoding steps was considered to be a separate and
qualitatively distinct performance metric. This makes sense when
the communication is long-range, since the ``exchange rate'' between
transmitter power and the power that ends up being delivered to the
receiver is very poor due to distance-induced attenuation.

\vspace{0.05in}

{\em In light of the advances in digital circuits and the need for
shorter-range communication, we propose a new way of formalizing what
it means for a coding approach to be ``capacity achieving'' using the
single natural metric: power.}

\subsection{Definitions} 
Assume the traditional information-theoretic model (see
e.g.~\cite{coverthomas, Gallager}) of fixed-rate discrete-time
communication with $k$ total information bits, $m$ channel uses, and
the rate of $R = \frac{k}{m}$ bits per channel use. As is traditional,
the rate $R$ is held constant while $k$ and $m$ are allowed to become
asymptotically large. $\pei$ is the average probability of bit error
on the $i$-th message bit and $\pe = \frac{1}{k} \sum_i \pei$ is used
to denote the overall average probability of bit error. No
restrictions are assumed on the codebooks aside from those required by
the channel model. The channel model is assumed to be indexed by the
power used in transmission. The encoder and decoder are assumed to be
physical entities that consume power according to some model that can
be different for different codes.

Let $\xi_T P_T$ be the actual power used in transmission and let $P_C$
and $P_D$ be the power consumed in the operation of the encoder and
decoder respectively. $\xi_T$ is the exchange rate (total path-loss)
that connects the power spent at the transmitter to the received power
$P_T$ that shows up at the receiver. In the spirit of
\cite{Vasudevan}, we assume that the goal of the system designer is to
minimize some weighted combination $P_{total} = \xi_T P_T + \xi_C P_C
+ \xi_D P_D$ where the vector $\vec{\xi} > 0$. The weights can be
different depending on the application\footnote{For example, in an
  RFID application, the power used by the tag is actually supplied
  wirelessly by the reader. If the tag is the decoder, then it is
  natural to make $\xi_D$ even larger than $\xi_T$ in order to account
  for the inefficiency of the power transfer from the reader to the
  tag. One-to-many transmission of multicast data is another example
  of an application that can increase $\xi_D$. The $\xi_D$ in that
  case should be increased in proportion to the number of receivers
  that are listening to the message.} and $\xi_T$ is tied to the
distance between the transmitter and receiver as well as the
propagation environment.

For any rate $R$ and average probability of bit error $\pe > 0$, we
assume that the system designer will minimize the weighted combination
above to get optimized $P_{total}(\vec{\xi}, \pe, R)$ as well as
constituent $P_T(\vec{\xi}, \pe, R), P_C(\vec{\xi}, \pe, R),$ and
$P_D(\vec{\xi}, \pe, R)$.

\begin{definition} \label{def:certainty}
The {\em certainty} of a particular encoding and decoding system is
the reciprocal of the average probability of bit error.  
\end{definition}
\vspace{0.1in}

\begin{definition} \label{def:weakcertaintyachieving}
An encoding and decoding system at rate $R$ bits per second is {\em
weakly certainty achieving} if $\liminf_{\pe \rightarrow 0} P_T(\vec{\xi},
\pe, R) < \infty$ for all weights $\vec{\xi} > 0$.
\end{definition}
\vspace{0.1in}

If an encoder/decoder system is not weakly certainty achieving, then
this means that it does not deliver on the revolutionary promise of
the Shannon waterfall curve from the perspective of transmit
power. Instead, such codes encourage system designers to pay for
certainty using unbounded transmission power.

\begin{definition} \label{def:strongcertaintyachieving}
An encoding and decoding system at rate $R$ bits per second is {\em
strongly certainty achieving} if $\liminf_{\pe \rightarrow 0}
P_{total}(\vec{\xi}, \pe, R) \neq \infty$ for all weights $\vec{\xi} >
0$. 
\end{definition}
\vspace{0.1in}

A strongly certainty-achieving system would deliver on the full spirit
of Shannon's vision: that certainty can be approached at finite total
power just by accepting longer end-to-end delays and amortizing the
total energy expenditure over many bits. The general distinction
between strong and weak certainty-achieving systems relates to how the
decoding power $P_D(\vec{\xi}, \pe, R)$ varies with the probability of
bit-error $\pe$ for a fixed rate $R$. Does it have waterfall or
waterslide behavior? For example, it is clear that uncoded
transmission has very simple encoding/decoding\footnote{All that is
  required is the minimum power needed to sample the received signal
  and threshold the result.} and so $P_D(\vec{\xi}, \pe, R)$ has a
waterfall behavior.

\begin{definition} \label{def:capacityachieving}
A \{weakly$|$strongly\} certainty-achieving system at rate $R$ bits per
second is also {\em \{weakly$|$strongly\} capacity achieving} if
\begin{equation} \label{eqn:ourcapacityachieving}
\liminf_{\xi_C, \xi_D \rightarrow \vec{0}} \liminf_{\pe \rightarrow 0}
P_T(\vec{\xi}, \pe, R) = C^{-1}(R)
\end{equation}
where $C^{-1}(R)$ is the minimum transmission power that is
predicted by the Shannon capacity of the channel model.
\end{definition}
\vspace{0.1in}

This sense of capacity achieving makes explicit the sense in which we
should consider encoding and decoding to be {\em asymptotically free},
but not actually free. The traditional approach of modeling encoding
and decoding as being actually free can be recovered by swapping the
order of the limits in \eqref{eqn:ourcapacityachieving}.

\begin{definition} \label{def:traditionalcapacityachieving}
An encoding and decoding system is considered {\em traditionally
  capacity achieving} if 
\begin{equation} \label{eqn:oldcapacityachieving}
\liminf_{\pe \rightarrow 0} \liminf_{\xi_C, \xi_D \rightarrow \vec{0}}
P_T(\vec{\xi}, \pe, R) = C^{-1}(R). 
\end{equation} 
where $C^{-1}(R)$ is the minimum transmission power that is predicted
by the Shannon capacity of the channel model.
\end{definition}
\vspace{0.1in}

By taking the limit $(\xi_C, \xi_D) \rightarrow 0$ for a fixed
probability of error, this traditional approach makes it impossible to
capture any fundamental tradeoff with complexity in an asymptotic
sense.

The conceptual distinction between the new
\eqref{eqn:ourcapacityachieving} and old
\eqref{eqn:oldcapacityachieving} senses of capacity-achieving systems
parallels Shannon's distinction between zero-error capacity and
regular capacity \cite{ShannonZeroError}. If $C(\epsilon, d)$ is the
maximum rate that can be supported over a channel using end-to-end
delay $d$ and average probability of error $\epsilon$, then
traditional capacity $C = \lim_{\epsilon \rightarrow 0} \lim_{d
  \rightarrow \infty} C(\epsilon, d)$ while zero-error capacity $C_0 =
\lim_{d \rightarrow \infty} \lim_{\epsilon \rightarrow 0} C(\epsilon,
d)$. When the limits are taken together in some balanced way, then we
get concepts like anytime capacity \cite{ControlPartI,
  OurUpperboundPaper}. It is known that $C_0 < C_{any} < C$ in
general and so it is natural to wonder whether any codes are capacity
achieving in the new stricter sense of
Definition~\ref{def:capacityachieving}.

\subsection{Are classical codes capacity achieving?}


\subsubsection{Dense linear block codes with nearest-neighbor decoding} 

Dense linear fixed-block-length codes are traditionally capacity
achieving under ML decoding \cite{Gallager}. To understand whether
they are weakly certainty achieving, we need a model for the encoding
and decoding power. Let $m$ be the block length of the code. Each
codeword symbol requires $mR$ operations to encode and it is
reasonable to assume that each operation consumes some energy. Thus,
the encoding power is $O(m)$. Meanwhile, a straightforward
implementation of ML (nearest-neighbor) decoding has complexity 
exponential in the block-length and thus it is reasonable to assume
that it consumes an exponential amount of power as well.

The probability of error for ML decoding drops exponentially with $m$
with an exponent that is bounded above by the sphere-packing exponent
$E_{sp}(R)$ \cite{Gallager}. An exponential reduction in the
probability of error is thus paid for using an exponential increase in
decoding power. Consequently, it is easy to see that the certainty
return on investments in decoding power is only polynomial. Meanwhile,
the certainty return on investments in transmit power is exponential
even for uncoded transmission. So no matter what the values are for
$\xi_{D} > 0$, in the high-certainty limit of very low probabilities
of error, an optimized communication system built using dense
linear block codes will be investing ever increasing amounts in
transmit power.

A plot of the resulting waterslide curves for both transmit power and
decoding power are given in
Figure~\ref{fig:waterslideblockML}. Following tradition, the
horizontal axes in the plots are given in normalized SNR units for
power. Notice how the optimizing system invests heavily in additional
transmit power to approach low probabilities of error.

\begin{figure}[htb]
\begin{center}
\includegraphics[scale=0.7]{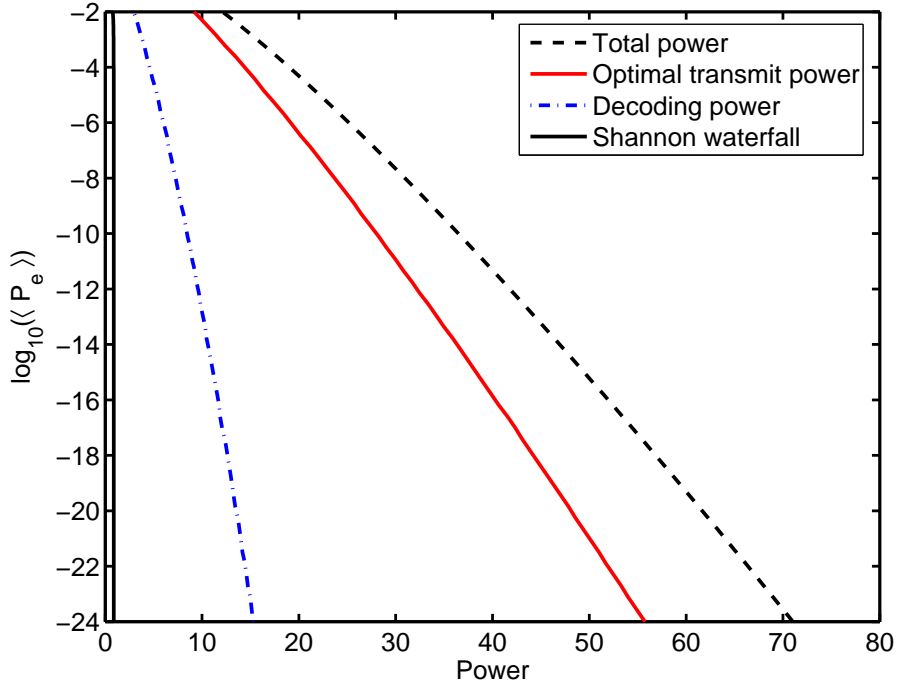}
\caption{The waterslide curves for transmit power, decoding power, and
  the total power for dense linear block-codes of rate $R=1/3$
  under brute-force ML decoding. It is assumed that the normalized energy
  required per operation at the decoder is $E=0.3$ and that it takes
  $2^{mR}\times mR$ operations per channel output to decode using
  nearest-neighbor search for a block length of $m$
  channel uses.}   
\label{fig:waterslideblockML}
\end{center}
\end{figure}

\subsubsection{Convolutional codes under Viterbi decoding} 
For convolutional codes, there are two decoding algorithms, and hence
two different analyses. (See \cite{ForneyML, ForneySeq} for details)
For Viterbi decoding, the complexity per-bit is exponential in the
constraint length $R L_c$ bits. The error exponents with the constraint
length of $L_c$ channel uses are upper-bounded
in~\cite{viterbi}, and this bound is given parametrically by
\begin{equation} \label{eqn:convolutionalexponent}
E_{conv}(R, P_T) = E_0(\rho, P_T) \,\,;\,\,  R = \frac{E_0(\rho, P_T)}{\rho}
\end{equation}
where $E_0$ is the Gallager function \cite{Gallager} and $\rho >
0$. The important thing here is that just as in dense linear block
codes, the certainty return on investments in decoding power is only
polynomial, albeit with a better polynomial than linear block-codes
since $E_{conv}(R,P_T)$ is higher than the sphere-packing bound for
block codes \cite{Gallager}. Thus, an optimized communication system
built using Viterbi decoding will also be investing ever increasing
amounts in transmit power. Viterbi decoding is not weakly certainty
achieving.

A plot of the resulting waterslide curves for both transmit power and
decoding power is given in Figure~\ref{fig:waterslideviterbi}. Notice
that the performance in Figure~\ref{fig:waterslideviterbi} is better
than that of Figure~\ref{fig:waterslideblockML}. This reflects the
superior error exponents of convolutional codes with respect to their
computational parameter --- the constraint length. 

\begin{figure}[htb]
\begin{center}
\includegraphics[scale=0.7]{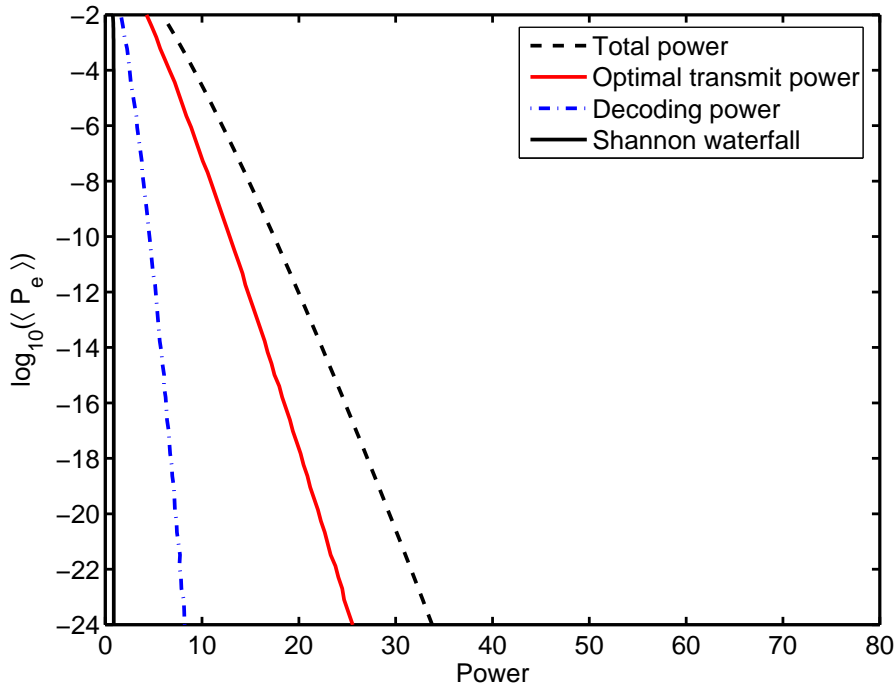}
\caption{The waterslide curves for transmit power, decoding power, and
  the total power for convolutional codes of rate $R=1/3$ used with
  Viterbi decoding. It is assumed that the normalized energy required per
  operation at the decoder is $E=0.3$ and that it takes
  $2^{L_cR}\times L_cR$ operations per channel output to decode using
  Viterbi search for a constraint length of $L_c$ channel uses.}
\label{fig:waterslideviterbi}
\end{center}
\end{figure}

\subsubsection{Convolutional codes under magical sequential decoding} \label{sec:firstmagic}

For convolutional codes with sequential decoding, it is shown
in~\cite{JacobsBerlekamp} that the average number of guesses must
increase to infinity if the message rate exceeds the \textit{cut-off
  rate}, $E_0(1)$. However, below the cut-off rate, the average number
of guesses is finite. Each guess at the decoder costs $L_c R$
multiply-accumulates and we assume that this means that average
decoding power also scales as $O(L_c)$ since at least one guess is
made for each received sample.

For simplicity, let us ignore the issue of the cut-off rate and
further assume that the decoder magically makes just one guess and
always gets the ML answer. The convolutional coding error exponent
\eqref{eqn:convolutionalexponent} still applies, and so the system's
certainty gets an exponential return for investments in decoding
power. It is now no longer obvious how the optimized-system will
behave in terms of transmit power.

For the magical system, the encoder power and decoder power are both
linear in the constraint-length. Group them together with the  
path-loss and normalize units to get a single effective term $\gamma
L_C$. The goal now is to minimize
\begin{equation} \label{eqn:magicperformance}
P_T + \gamma L_c 
\end{equation}
over $P_T$ and $L_C$ subject to the probability of error constraint
that $\ln \frac{1}{\pe} = E_{conv}(R,P_T) \frac{L_c}{R}$. Since we are
interested in the limit of $\ln \frac{1}{\pe} \rightarrow \infty$, it
is useful to turn this around and use Lagrange multipliers. A little
calculation reveals that the optimizing values of $P_T$ and $L_c$ must
satisfy the balance condition
\begin{equation} \label{eqn:magicbalance}
E_{conv}(R,P_T) = \gamma L_C \frac{\partial E_{conv}(R, P_T)}{\partial P_T}
\end{equation}
and so (neglecting integer-effects) the optimizing constraint-length
is either $1$ (uncoded transmission) or 
\begin{equation} \label{eqn:magicoptimalconstraint}
L_c = \frac{1}{\gamma} E_{conv}(R,P_T)/\frac{\partial E_{conv}(R, P_T)}{\partial P_T}.
\end{equation}

To get ever lower values of $\pe$, the transmit power $P_T$ must
therefore increase unboundedly unless the ratio
$E_{conv}(R,P_T)/\frac{\partial E_{conv}(R, P_T)}{\partial P_T}$
approaches infinity for some finite $P_T$. Since the convolutional
coding error exponent \eqref{eqn:convolutionalexponent} does not go to
infinity at a finite power, this requires $\frac{\partial E_{conv}(R,
  P_T)}{\partial P_T}$ to approach zero. For AWGN style channels, this
only occurs\footnote{There is a slightly subtle issue here. Consider
  random codes for a moment. The convolutional random-coding error
  exponent is flat at $E_0(1,P_T)$ for rates $R$ below the
  computational cutoff rate. However, that flatness with rate $R$ is
  not relevant here. For any fixed constellation, the $E_0(1,P_T)$ is
  a strictly monotonically increasing function of $P_T$, even though
  it asymptotes at a non-infinite value. This is not enough since the
  derivative with transmit power still tends to zero only as $P_T$
  goes to infinity.}  as $P_T$ approaches infinity and thus the gap
between $R$ and the capacity gets large.

\begin{figure}[htb]
\begin{center}
\includegraphics[scale=0.7]{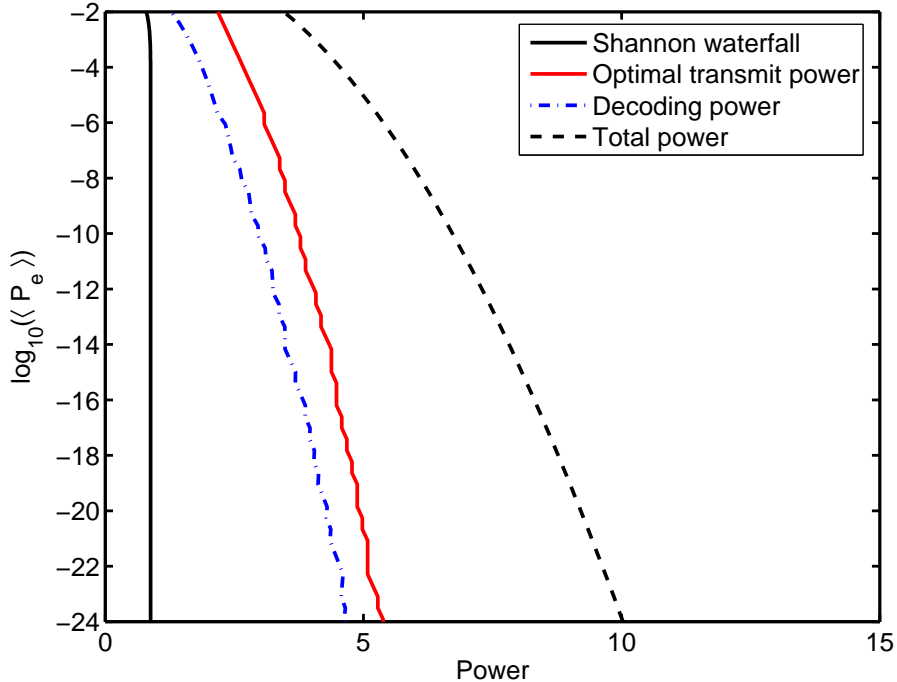}
\caption{The waterslide curves for transmit power, decoding power, and
  the total power for convolutional codes of rate $R=1/3$ used with
  ``magical'' sequential decoding. It is assumed that the normalized energy
  required per operation at the decoder is $E=0.3$ and that the 
  decoding requires just $L_cR$ operations per channel output.} 
\label{fig:waterslidemagic}
\end{center}
\end{figure}

The resulting plots for the waterslide curves for both transmit power
and encoding/decoding power are given in
Figure~\ref{fig:waterslidemagic}. Although these plots are much better
than those in Figure~\ref{fig:waterslideviterbi}, the surprise is that
even such a magical system that attains an error-exponent with
investments in decoding power is unable to be weakly certainty
achieving at any rate. Instead, the optimizing transmit power goes to
infinity.

\subsubsection{Dense linear block codes with magical syndrome
    decoding}

It is well known that linear codes can be decoded by looking at the
syndrome of the received codeword \cite{Gallager}. Suppose that we had
a magical syndrome decoder that could use a free lookup table to
translate the syndrome into the ML corrections to apply to the
received codeword. The complexity of the decoding would just be the
complexity of computing the syndrome. For a dense random linear block
code, the parity-check matrix is itself typically dense and so the
per-channel-output complexity of computing each bit of the syndrome is
linear in the block-length. This gives rise to behavior like that of
magical sequential decoding above and is illustrated in
Figure~\ref{fig:waterslideblockmagic}.

\begin{figure}[htb]
\begin{center}
\includegraphics[scale=0.7]{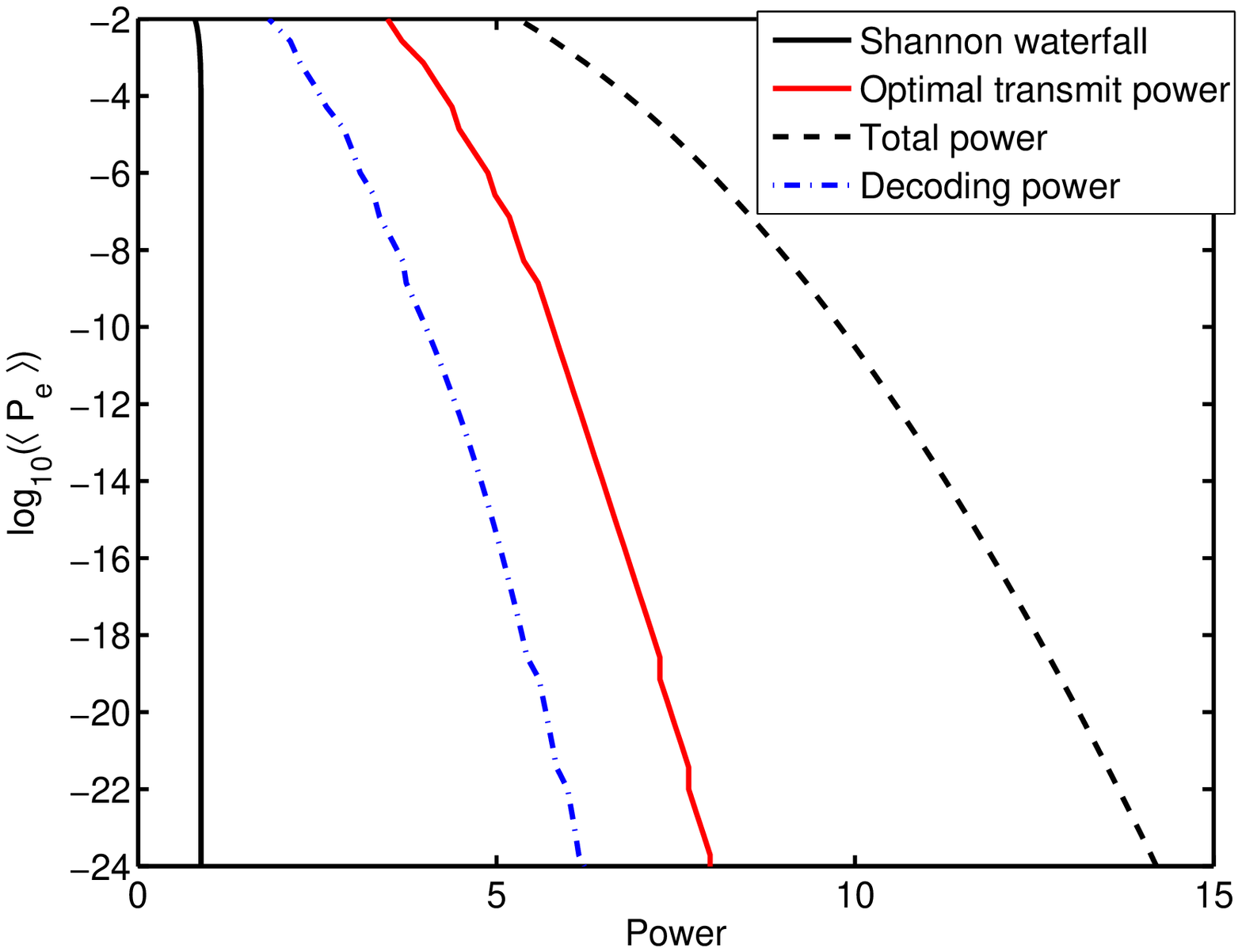}
\caption{The waterslide curves for transmit power, decoding power, and
  the total power for dense linear block-codes of rate $R=1/3$
  under magical syndrome decoding. It is assumed that the normalized
  energy required per operation at the decoder is $E=0.3$ and that the 
  decoding requires just $(1-R) mR$ operations per channel output to
  compute the syndrome.} 
\label{fig:waterslideblockmagic}
\end{center}
\end{figure}

\vspace{0.1in}

From the above discussion, it seems that in order to have even a
weakly certainty-achieving system, the certainty-return for
investments in encoding/decoding power must be faster than
exponential!

\section{Parallel iterative decoding: a new hope} \label{sec:sysmod}

The unrealistic magical syndrome decoder suggests a way forward. If
the parity-check matrix were sparse, then it would be possible to
compute the syndrome using a constant number of operations per
received symbol. If the probability of error dropped with
block-length, that would give rise to an infinite-return on
investments in decoder power. This suggests looking in the direction
of LDPC codes \cite{gallagerThesis}. While magical syndrome decoding
is unrealistic, many have observed that message-passing decoding gives
good results for such codes while being implementable
\cite{ModernCodingTheory}.

Upon reflection, it is clear that parallel iterative decoding based on
message passing holds out the potential for {\em super-exponential}
improvements in probability of error with decoding power. This is
because messages can reach an exponential-sized neighborhood in only a
small number of iterations, and large-deviations thinking suggests
that there is the possibility for an exponential reduction in the
probability of error with neighborhood size. In fact, exactly this
sort of double-exponential reduction in the probability of error under
iterative decoding has been shown to be possible for regular LDPCs
\cite[Theorem 5]{Lentmaier05}.

To make all this precise, we need to fix our model of the problem and
of an implementable decoder. Consider a point-to-point communication
link. An information sequence $\mk{B}$ is encoded into $2^{mR}$ codeword
symbols $\m{X}$, using a possibly randomized encoder.  The observed
channel output is $\m{Y}$. The information sequences are assumed to
consist of iid fair coin tosses and hence the rate of the code is
$R=k/m$. Following tradition, both $k$ and $m$ are considered to be
very large. We ignore the complexity of doing the encoding under the
hope that encoding is simpler than decoding.\footnote{For certain
  LDPC-codes, it is shown in \cite{RichardsonUrbankeEncoding} that
  encoding can be made to have complexity linear in the block-length
  for a certain model of encoding. In our context, linear complexity
  means that the complexity per data bit is constant and thus this
  does not require power at the encoder that grows with either the
  block length or the number of decoder iterations. We have not yet
  verified if the complexity of encoding is linear under our
  computational model.}

Two channel models are considered: the BSC and the power-constrained
AWGN channel. The true channel is always denoted $P$. The underlying
AWGN channel has noise variance $\sigma_P^2$ and the average received
power is denoted $P_T$ so the received SNR is
$\frac{P_T}{\sigma_P^2}$. Similarly, we assume that the BSC has
crossover probability $p$. We consider the BSC to have resulted from
BPSK modulation followed by hard-decision detection on the AWGN
channel and so $p =
\mathcal{Q}\left(\sqrt{\frac{P_T}{\sigma_P^2}}\right)$.

For maximum generality, we do not impose any {\em a priori} structure
on the code itself. Instead, inspired by \cite{TarokhVLSI, Shanbhag,
  HotChannel1, HotChannel2}, we focus on the parallelism of the
decoder and the energy consumed within it. We assume that the decoder
is physically made of computational nodes that pass messages to each
other in parallel along physical (and hence unchanging) wires. A
subset of nodes are designated `message nodes' in that each is
responsible for decoding the value of a particular message bit.
Another subset of nodes (not necessarily disjoint) has members that
are each initialized with at most one observation of the received
channel-output symbols. There may be additional computational nodes
that are just there to help decode.

The implementation technology is assumed to dictate that each
computational node is connected to at most $\alpha+1 > 2$ other
nodes\footnote{In practice, this limit could come from the number of
  metal layers on a chip. $\alpha = 1$ would just correspond to a big
  ring of nodes and is uninteresting for that reason.} with
bidirectional wires. No other restriction is assumed on the topology
of the decoder. In each iteration, each node sends (possibly
different) messages to all its neighboring nodes. {\bf No restriction
  is placed on the size or content of these messages except for the
  fact that they must depend on the information that has reached the
  computational node in previous iterations.} If a node wants to
communicate with a more distant node, it has to have its message
relayed through other nodes. No assumptions are made regarding the
presence or absence of cycles in this graph. The neighborhood size at
the end of $l$ iterations is denoted by $n \leq \alpha^{l+1}$. We
assume $m\gg n$. Each computational node is assumed to consume a fixed
$E_{node}$ joules of energy at each iteration.

Let the average probability of bit error of a
code be denoted by $\pe_P$ when it is used over channel $P$. The goal
is to derive a lower bound on the neighborhood size $n$ as a function
of $\pe_P$ and $R$. This then translates to a lower bound on the
number of iterations which can in turn be used to lower bound the
required decoding power.

Throughout this paper, we allow the encoding and decoding to be
randomized with all computational nodes allowed to share a common pool
of common randomness. We use the term `average probability of error'
to refer to the probability of bit error averaged over the channel
realizations, the messages, the encoding, and the decoding.

\section{Lower bounds on decoding complexity: iterations and power}
\label{sec:lowbound}
In this section, lower bounds are stated on the computational
complexity for iterative decoding as a function of the gap from
capacity. These bounds reveal that the decoding neighborhoods must
grow unboundedly as the system tries to approach capacity. We assume
the decoding algorithm is implemented using the iterative technology
described in Section~\ref{sec:sysmod}. The resulting bounds are then
optimized numerically to give plots of the optimizing transmission and
decoding powers as the average probability of bit error goes to
zero. For transmit power, it is possible to evaluate the limiting
value as the system approaches certainty. However, decoding power is
shown to diverge to infinity for the same limit. This shows that the
lower bound does not rule out weakly capacity-achieving schemes, but
strongly capacity-achieving schemes are impossible using
Section~\ref{sec:sysmod}'s model of iterative decoding.

\subsection{Lower bounds on the probability of error in terms of
  decoding neighborhoods} \label{sec:basicbounds}

The main bounds are given by theorems that capture a local
sphere-packing effect. These can be turned around to give a family of
lower bounds on the neighborhood size $n$ as a function of
$\pe_P$. This family is indexed by the choice of a hypothetical
channel $G$ and the bounds can be optimized numerically for any
desired set of parameters. 

\begin{theorem} \label{thm:basicBSCbound}
Consider a BSC with crossover probability $p < \frac{1}{2}$. Let $n$
be the maximum size of the decoding neighborhood of any individual
bit. The following lower bound holds on the average probability of bit
error. 
\begin{equation}
\label{eq:peip}
\pe_P \geq  \sup_{C^{-1}(R) < g \leq \frac{1}{2}}\frac{h_b^{-1}(\delta(G))}{2}
2^{-nD\left(g||p\right)}\left(\frac{p(1-g)}{g(1-p)}\right)^{\epsilon \sqrt{n}}
\end{equation}
where $h_b(\cdot{})$ is the usual binary entropy
function, $D(g||p) = g\lo{\frac{g}{p}} + (1-g)\lo{\frac{1-g}{1-p}}$ is
the usual KL-divergence, and  
\begin{eqnarray}
\delta(G)&=& 1- \frac{C(G)}{R} \label{eqn:deltadefinition}\\
\text{where } C(G) &=& 1 - h_b(g) \nonumber\\
\text{and }\epsilon &=& \sqrt{\frac{1}{K(g)}\lo{\frac{2}{h_b^{-1}(\delta(G))}}} \label{eqn:epsilondef}\\
\text{where }K(g) &=& \inf_{0<\eta<1-g} \frac{D(g+\eta ||g)}{\eta^2}. \label{eqn:Kdef}
\end{eqnarray}
\end{theorem}
\begin{proof}
See Appendix~\ref{app:peip}.
\end{proof}
\vspace{0.1in}

\begin{theorem}
\label{thm:basicAWGNbound}
For the AWGN channel and the decoder model in
Section~\ref{sec:sysmod}, let $n$ be the  
maximum size of the decoding neighborhood of any individual
message bit. The following lower bound holds on the average
probability of bit error.
\begin{equation} \label{eq:lbnawgn}
\pe_P \geq  \sup_{\sigma_G^2:\;C(G) < R}
\frac{h_b^{-1}\left(\delta(G)\right)}{2}
\exp\left(-nD(\sigma_G^2||\sigma_P^2)
- \sqrt{n} \left(\frac{3}{2} + 2 \lon{\frac{2}{h_b^{-1}(\delta(G))}}\right) 
\left(\frac{\sigma_G^2}{\sigma_P^2}-1\right)\right)
\end{equation}
where $\delta(G) = 1-C(G)/R$, the capacity
$C(G)=\frac{1}{2}\lo{1+\frac{P_T}{\sigma_G^2}}$, and the KL divergence
$D(\sigma_G^2||\sigma_P^2) = \frac{1}{2}\left[\frac{\sigma_{G}^2}{\sigma_P^2}-1
    -\ln\left(\frac{\sigma_{G}^2}{\sigma_P^2}\right) \right]$. 

The following lower bound also holds on the average probability of bit
error
\begin{equation} \label{eq:lbnawgnNumeric}
\pe_P \geq  \sup_{\sigma_G^2 > \sigma_P^2 \mu(n) :\;C(G) < R}
\frac{h_b^{-1}\left(\delta(G)\right)}{2}
\exp\left(-nD(\sigma_G^2||\sigma_P^2)
- \frac{1}{2} \phi(n, h_b^{-1}\left(\delta(G)\right)) \left(\frac{\sigma_G^2}{\sigma_P^2}-1\right)\right),
\end{equation}
where 
\begin{eqnarray} 
 \mu(n) &=& \frac{1}{2}(1 + \frac{1}{T(n)+1} + \frac{4T(n) +
   2}{nT(n)(1 + T(n))}) \label{eqn:mudef}\\
\text{where }
    T(n)&=& -W_L(-\exp(-1)(1/4)^{1/n}) \label{eqn:Tdef} \\
\text{and }
    W_L(x) \text{ solves }x &=& W_L(x)\exp(W_L(x))  \label{eqn:Wdef} \\
\text{while satisfying }W_L(x) &\leq& -1 \mbox{ }\forall x \in
    [-\exp(-1),0], \nonumber
\end{eqnarray}
and 
\begin{equation}
\phi(n,y) = -n (W_L\left(-\exp(-1)(\frac{y}{2})^{\frac{2}{n}}\right)+1).
\label{eqn:phidef}
\end{equation}
The $W_L(x)$ is the transcendental Lambert $W$ function
\cite{LambertWRef} that is defined implicitly by the relation \eqref{eqn:Wdef}
above.
\end{theorem}
\begin{proof}
See Appendix~\ref{app:peiawgn}.
\end{proof}

The expression \eqref{eq:lbnawgnNumeric} is better for plotting bounds
when we expect $n$ to be moderate while \eqref{eq:lbnawgn} is more
easily amenable to asymptotic analysis as $n$ gets large.

\subsection{Joint optimization of the weighted total power}
\label{sec:joint}
Consider the total energy spent in transmission. For transmitting $k$
bits at rate $R$, the number of channel uses is $m = k/R$. If each
transmission has power $\xi_T P_T$, the total energy used in
transmission is  $\xi_T P_T m$. 

At the decoder, let the number of iterations be $l$. Assume that each
node consumes $E_{node}$ joules of energy in each iteration. The
number of computational nodes can be lower bounded by the number $m$
of received channel outputs.
\begin{equation}
E_{dec}\geq E_{node}\times m\times l.
\end{equation} 
This gives a lower bound of $P_D \geq E_{node} l$ for decoder
power. There is no lower bound on the encoder complexity and so the
encoder is considered free.  This results in the following bound for
the weighted total power
\begin{equation}
P_{total}\geq \xi_T P_T + \xi_D E_{node} \times l.
\end{equation} 
Using $l \geq \frac{\lo{n}}{\lo{\alpha}}$ as the natural lower bound
on the number of iterations given a desired maximum neighborhood size,
\begin{eqnarray}
\nonumber P_{total} &\geq & \xi_T P_T +\frac{\xi_D E_{node} \lo{n}}{\lo{\alpha}}\\
& \propto & \frac{P_T}{\sigma_P^2} + \gamma \lo{n}
\label{eq:totavg}
\end{eqnarray} 
where $\gamma = \frac{\xi_D E_{node}}{\sigma_P^2 \xi_T \lo{\alpha}}$
is a constant that summarizes all the technology and environmental
terms. The neighborhood size $n$ itself can be lower bounded by
plugging the desired average probability of error into
Theorems~\ref{thm:basicBSCbound} and \ref{thm:basicAWGNbound}.

It is clear from~\eqref{eq:totavg} that for a given rate $R$ bits per
channel use, if the transmit power $P_T$ is extremely close to that
predicted by the channel capacity, then the value of $n$ would have to
be extremely large. This in turn implies that there are a large number
of iterations and thus it would require high power consumption at the
decoder. Therefore, the optimized encoder has to transmit at a power
larger than that predicted by the Shannon limit in order to decrease
the power consumed at the decoder. Also, from~\eqref{eq:peip}, as
$\pe\rightarrow 0$, the required neighborhood size $n \rightarrow
\infty$. This implies that for any fixed value of transmit power, the
power consumed at the decoder diverges to infinity as the probability
of error converges to zero. Hence the total power consumed must
diverge to infinity as the probability of error converges to
zero. This immediately rules out the possibility of having a strongly
certainty-achieving code using this model of iterative decoding. The
price of certainty is infinite power. The only question that remains
is whether the optimal transmitter power can remain bounded or not.

The optimization can be performed numerically once the exchange rate
$\xi_T$ is fixed, along with the technology parameters $E_{node},
\alpha, \xi_C, \xi_D$. Figures \ref{fig:waterfall} and
\ref{fig:waterslideawgn} show the total-power waterslide curves for
iterative decoding assuming the lower bounds.\footnote{The
  order-of-magnitude choice of $\gamma = 0.3$ was made using the
  following numbers. The energy cost of one iteration at one node
  $E_{node} \approx 1$pJ (optimistic extrapolation from the reported
  values in \cite{manishAllerton07, HowardSchlegel}), path-loss $\xi_T
  \approx 86 dB$ corresponding to a range in the tens of meters,
  thermal noise energy per sample $\sigma_P^2 \approx 4 \times
  10^{-21}$J from $kT$ with $T$ around room temperature, and
  computational node connectivity $\alpha = 4$.}  These plots show the
effect of changing the relative cost of decoding. The waterslide
curves become steeper as decoding becomes cheaper and the plotted
scale is chosen to clearly illustrate the double-exponential
relationship between decoder power and probability of error.

\begin{figure}[htb]
\begin{center}
\includegraphics[scale=0.7]{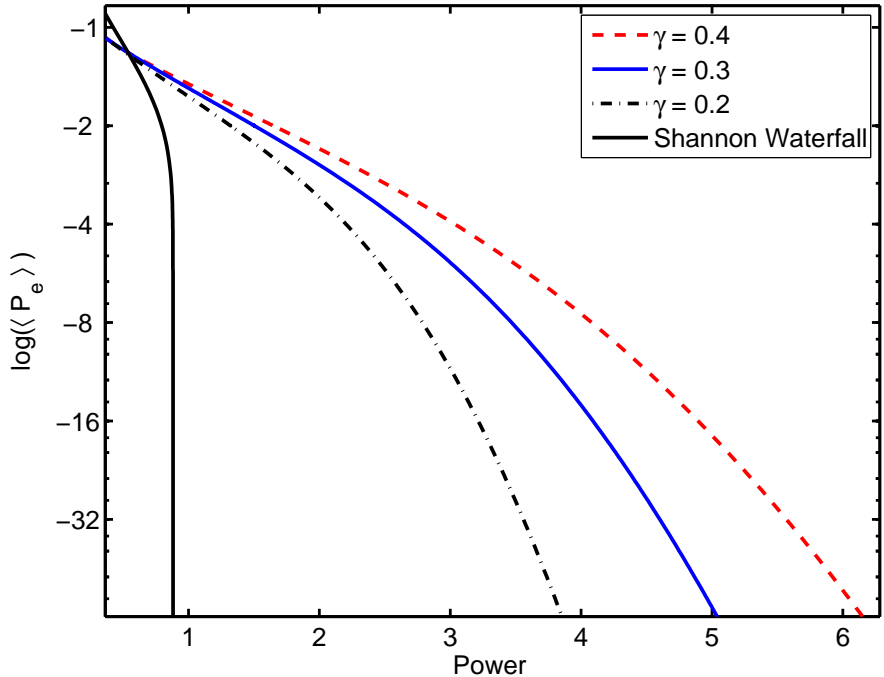}
\caption{The BSC Waterslides: plots of $\log(\pe)$ vs bounds on
  required total power for any fixed rate-$1/3$ code transmitted over
  an AWGN channel using BPSK modulation and hard decisions. $\gamma =
  \xi_D E_{node}/(\xi_T \sigma_P^2 \lo{\alpha})$ denotes the
  normalized energy per node per iteration in SNR units. Total power
  takes into account the transmit power as well as the power consumed
  in the decoder. The Shannon limit is a universal lower bound for all
  $\gamma$.}
\label{fig:waterfall}
\end{center}
\end{figure}

\begin{figure}[htb]
\begin{center}
\includegraphics[scale=0.7]{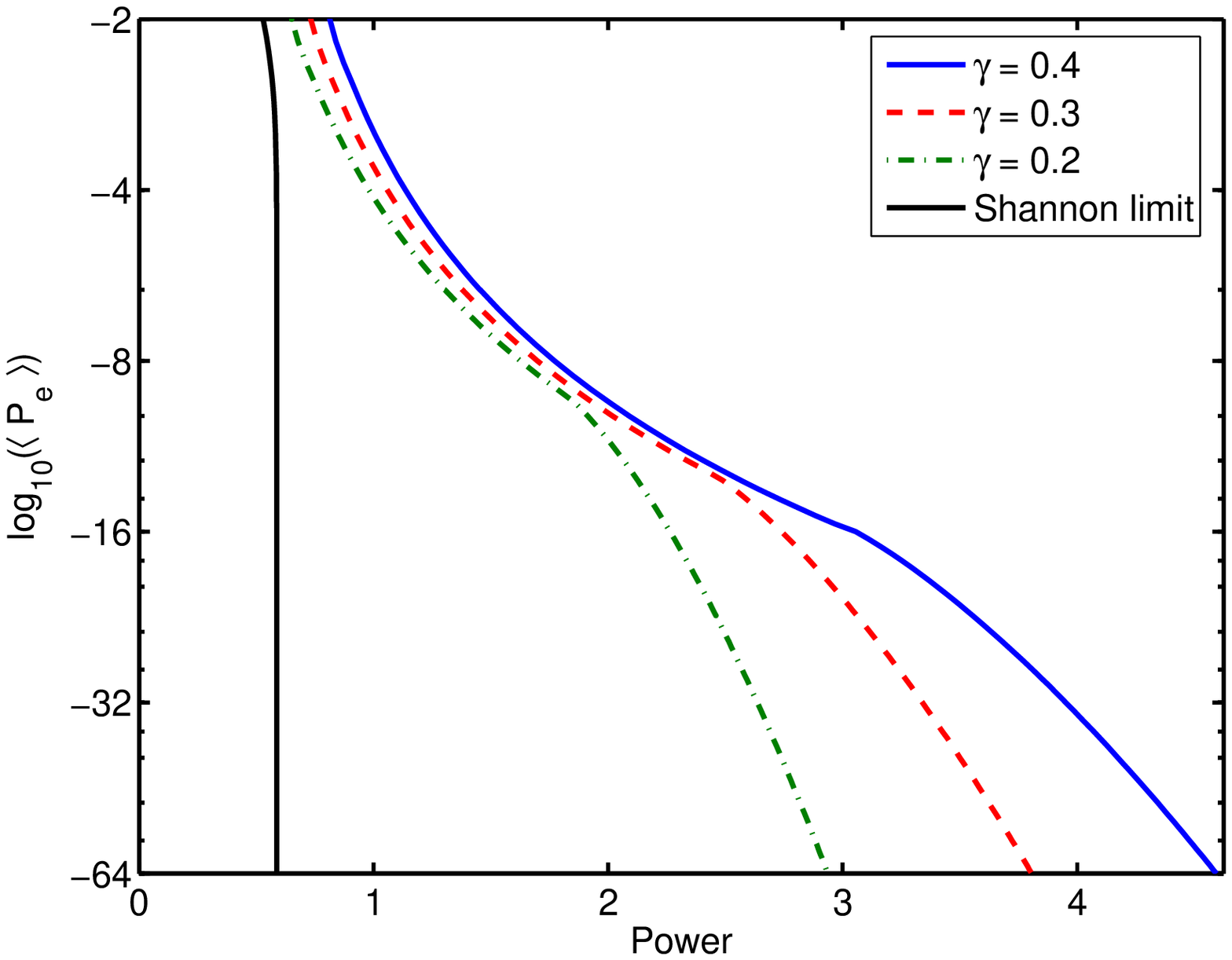}
\caption{The AWGN Waterslides: plots of $\log(\pe)$ vs bounds on
  required total power for any fixed rate-$1/3$ code transmitted over
  an AWGN channel. The initial segment where all the waterslide curves
  almost coincide illustrates the looseness of the bound since that
  corresponds to the case of $n = 1$ or when the bound suggests that
  uncoded transmission could be optimal. However, the probability of
  error is too optimistic for uncoded transmission.}
\label{fig:waterslideawgn}
\end{center}
\end{figure}

Figure~\ref{fig:upperbound} fixes the technology parameters and breaks
out the optimizing transmit power and decoder power as two separate
curves. It is important to note that only the weighted total power
curve is a true bound on what a real system could achieve. The
constituent $P_T$ and $P_D$ curves are merely indicators of what the
qualitative behaviour would be if the true tradeoff behaved like the
lower bound.\footnote{This doesn't mean that the bound is useless
  however. A lower bound on the transmit power can be computed once
  any implementable scheme exists. Simply look up where the
  bounded total power matches the implementable scheme. This will
  immediately give rise to lower bounds on the optimal transmit and
  decoding powers.} The optimal transmit power approaches a finite
limit as the probability of error approaches $0$. This limit can be
calculated directly by examining \eqref{eq:peip} for the BSC and
\eqref{eq:lbnawgn} for the AWGN case. 

\begin{figure}[htb]
\begin{center}
\includegraphics[scale=0.7]{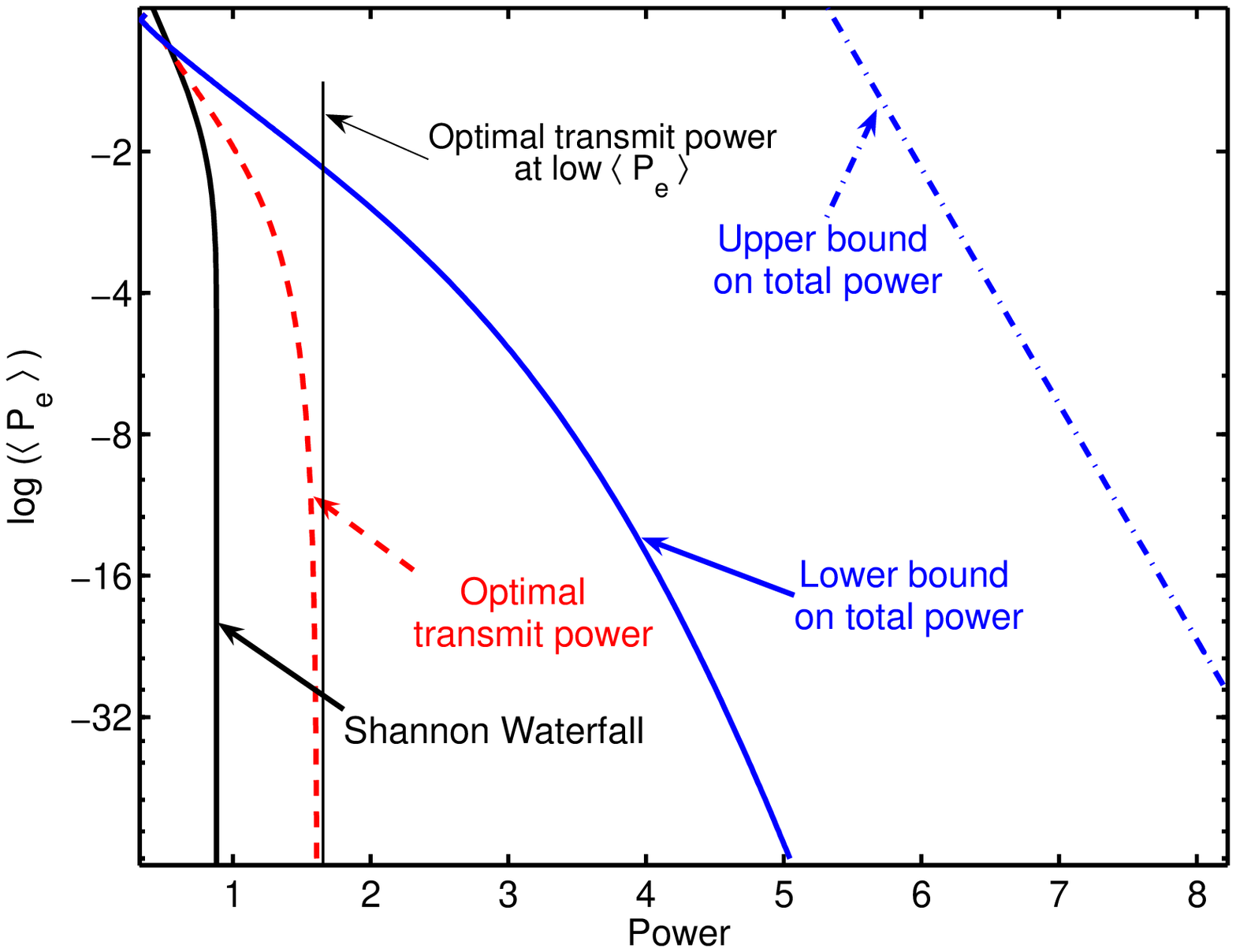}
\caption{The BSC Waterslide curve for $\gamma=0.3$, $R = 1/3$. An
  upper bound (from Section~\ref{sec:upbound}) that is parallel to the
  lower bound is also shown along with the heuristically optimal
  transmit power. This transmit power is larger than that predicted by
  the Shannon limit for small probabilities of error. This suggests
  that the transmitter has to make accommodations for the decoder
  complexity in order to minimize the total power consumption.}
\label{fig:upperbound}
\end{center}
\end{figure}

To compute this limit, recall that the goal is to optimize
$\frac{P_T}{\sigma_P^2} + \gamma \lo{n}$ over $P_T$ so as to satisfy a
probability of error constraint $\pe$, where the probability of error
is tending to zero. Instead of constraining the probability of error
to be small, it is just as valid to constrain $\gamma \log \log
\frac{1}{\pe}$ to be large. Now, take the logarithm of both sides of
\eqref{eq:peip} (or similarly for
\eqref{eq:lbnawgn}). It is immediately clear that the
only order $n$ term is the one that multiplies the divergence. Since
$n \rightarrow \infty$ as $\pe \rightarrow 0$, this term will dominate
when a second logarithm is taken. Thus, we know that the bound on the
double logarithm of the certainty $\gamma \log \log \frac{1}{\pe}
\rightarrow \gamma \lo{n} + \gamma \log f(R,\frac{P_T}{\sigma_P^2})$
where $f(R,\frac{P_T}{\sigma_P^2}) = D(G||P)$ is the divergence
expression involving the details of the channel. It turns out that $G$
approaches $C^{-1}(R)$ when $\pe \rightarrow 0$ since the divergence
is maximized there.

Optimizing for $\zeta = \frac{P_T}{\sigma_P^2}$ by taking derivatives
and setting to zero gives:
\begin{equation} \label{eqn:poweroffsetequation}
f(R,\zeta) / \frac{\partial f(R,\zeta)}{\partial \zeta} =
\gamma.
\end{equation}
It turns out that this has a unique root $\zeta(R,\gamma)$ for all
rates $R$ and technology factors $\gamma$ for both the BSC and the
AWGN channel. 

The key difference between \eqref{eqn:magicbalance} and
\eqref{eqn:poweroffsetequation} is that no term that is related to the
neighborhood-size or number of iterations has survived in
\eqref{eqn:poweroffsetequation}. This is a consequence of the
double-exponential\footnote{In fact, it is easy to verify that
  anything faster than double-exponential will also work.} reduction
in the probability of error with the number of iterations and the fact
that the transmit power shows up in the outer and not the inner
exponential.

To see if iterative decoding allows weakly capacity-achieving codes,
we take the limit of $\xi_D \rightarrow 0$ which implies $\gamma
\rightarrow 0$. \eqref{eqn:poweroffsetequation} then suggests that we
need to solve $f(R,\zeta) / \frac{\partial f(R,\zeta)}{\partial \zeta}
= 0$ which implies that either the numerator is zero or the
denominator becomes infinite. For AWGN or BSC channels, the slope of
the error exponent $f(R,\zeta)$ is monotonically decreasing as the SNR
$\zeta \rightarrow \infty$ and so the unique solution is where
$f(R,\zeta) = D(C^{-1}(R)||P_T) = 0$. This occurs at $P_T = C^{-1}(R)$
and so the lower bounds of this section do not rule out weakly
capacity-achieving codes.

In the other direction, as the $\gamma$ term gets large, the $P_T(R,
\gamma)$ increases. This matches the intuition that as the relative
cost of decoding increases, more power should be allocated to the
transmitter. This effect is plotted in
Figure~\ref{fig:gammavspower}. Notice that it becomes insignificant
when $\gamma$ is very small (long-range communication) but becomes
non-negligible whenever the $\gamma$ exceeds $0.1$.

\begin{figure}[htb]
\begin{center}
\includegraphics[scale=0.7]{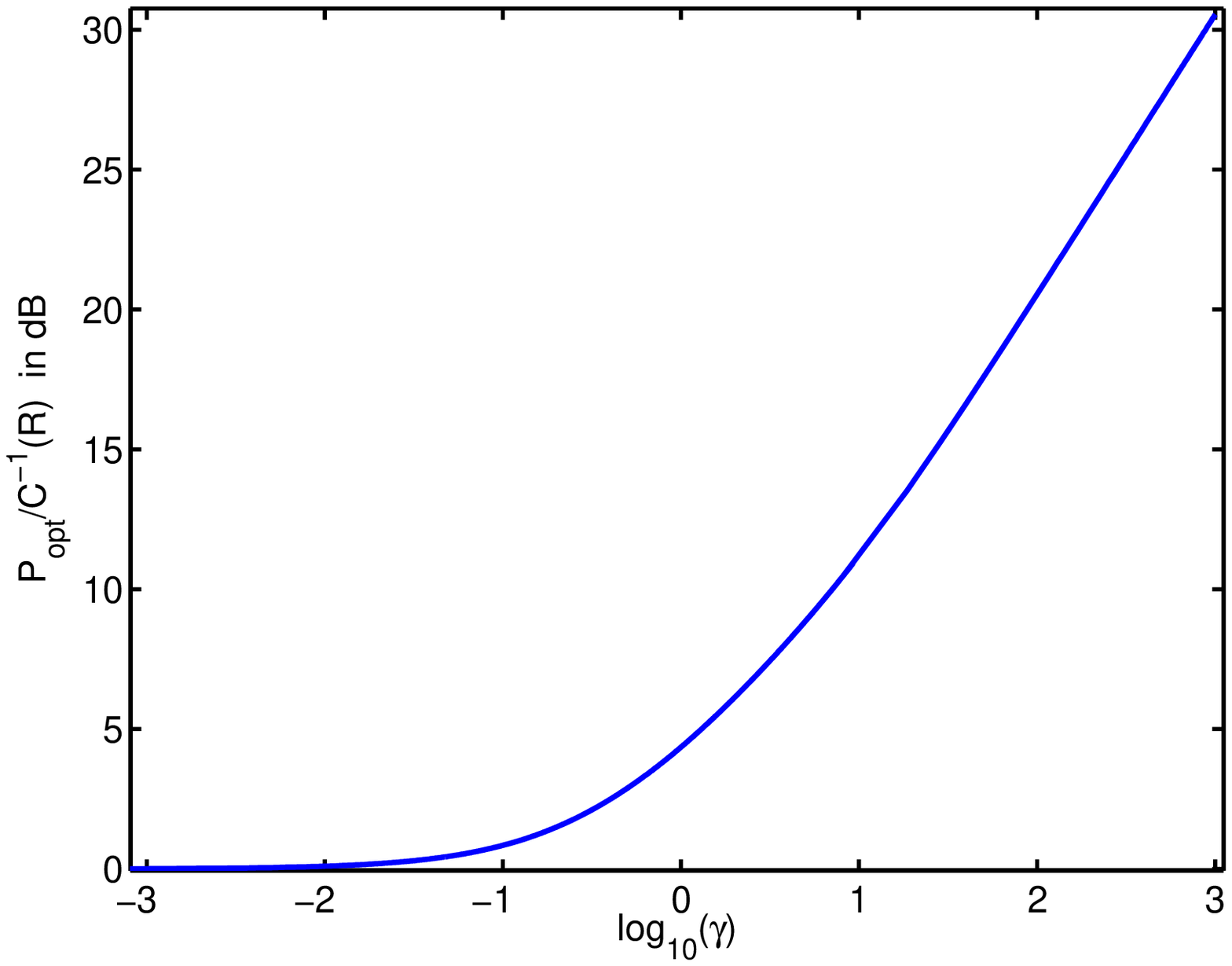}
\caption{The impact of $\gamma$ on the heuristically predicted optimum
  transmit power for the BSC used at $R=\frac{1}{3}$. The plot shows
  the gap from the Shannon prediction in a factor sense.}
\label{fig:gammavspower}
\end{center}
\end{figure}

Figure~\ref{fig:ratevspower} illustrates how the effect varies with
the desired rate $R$. The penalty for using low-rate codes is quite
significant and this gives further support to the lessons drawn from
\cite{massaad, goldsmithbahai} with some additional intuition
regarding why it is fundamental. The error exponent governing the
probability of error as a function of the neighborhood size is limited
by the sphere-packing bound at rate $0$ -- this is finite and the only
way to increase it is to pay more transmit power. However, the
decoding power is proportional to the number of received samples and
this is larger at lower rates.

\begin{figure}[htb]
\begin{center}
\includegraphics[scale=0.7]{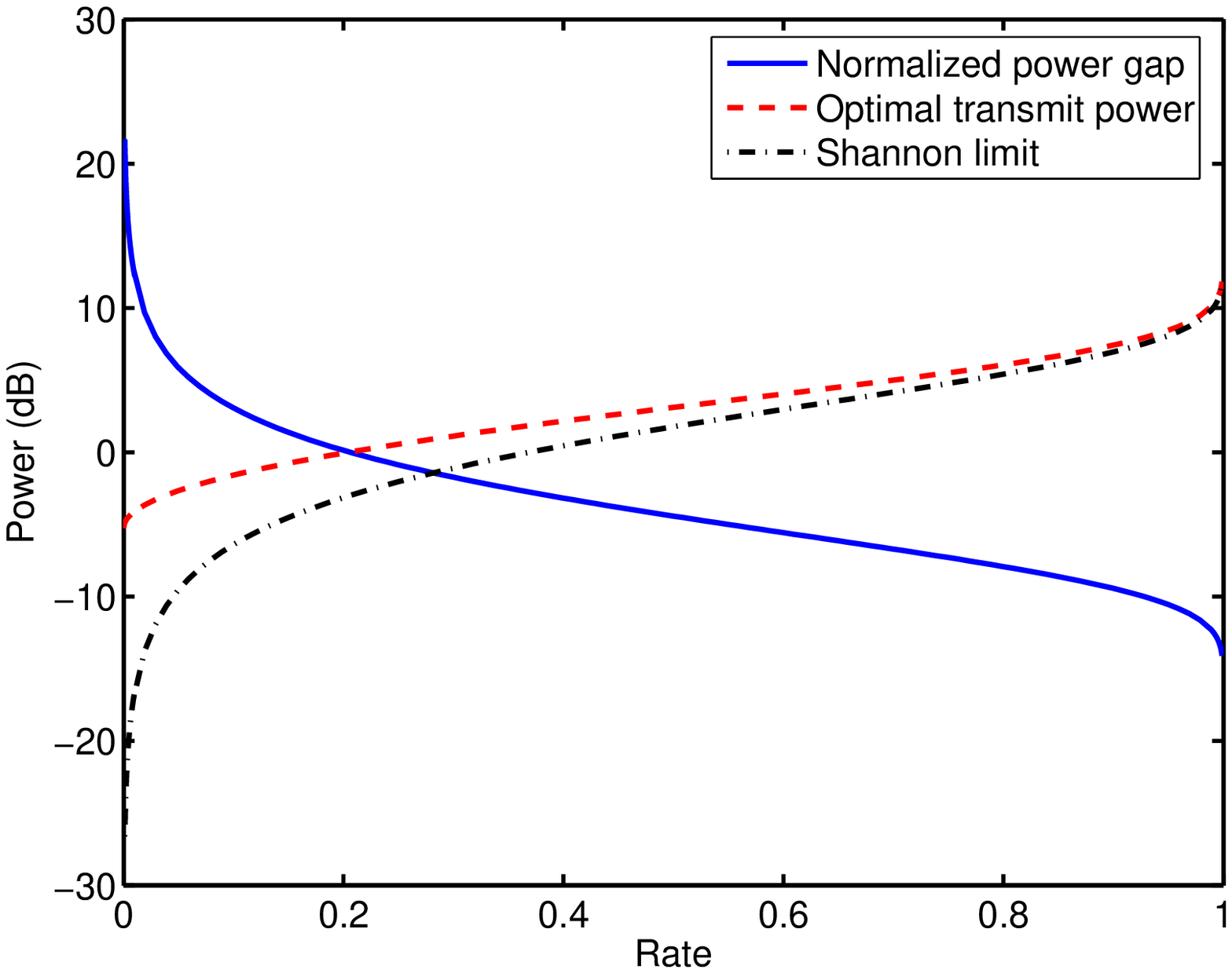}
\caption{The impact of rate $R$ on the heuristically predicted optimum
  transmit power for $\gamma = 0.3$. The plot shows the Shannon
  minimum power, our predictions, and the ratio of the difference between the 
  two to the Shannon minimum. Notice that the predicted extra power is
  very substantial at low data rates.}
\label{fig:ratevspower}
\end{center}
\end{figure}

Finally, the plots were all made assuming that the neighborhood size
$n$ could be chosen arbitrarily and the number of iterations could be
a real number rather than being restricted to integer values. This is
fine when the desired probability of error is low, but it turns out
that this integer effect cannot be neglected when the tolerable
probability of error is high. This is particularly significant when
$\gamma$ is large. To see this, it is useful to consider the boundary
between when uncoded transmission is optimal and when coding might be
competitive. This is done in Figure~\ref{fig:uncoded} where the
minimum $\gamma \lo{\alpha}$ power required for the first decoder
iteration is instead given to the transmitter. Once $\gamma > 10$, it
is hard to beat uncoded transmission unless the desired probability of
error is very low indeed.

\begin{figure}[htb]
\begin{center}
\includegraphics[scale=0.7]{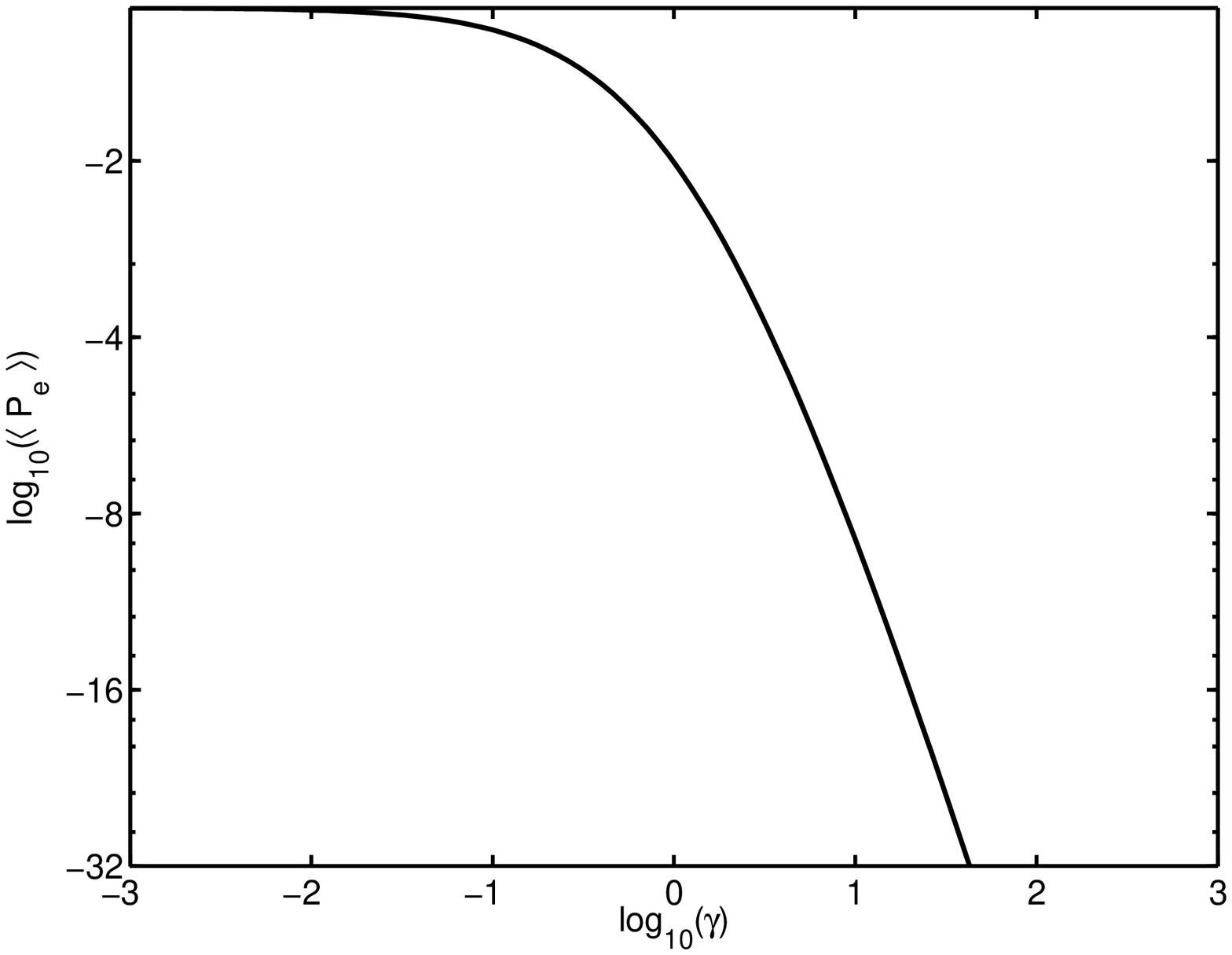}
\caption{The probability of error below which coding could potentially
  be useful. This plot assumes an AWGN channel used with BPSK
  signaling and hard-decision detection, target message rate
  $R=\frac{1}{3}$, and an underlying iterative-decoding architecture
  with $\alpha = 3$. This plot shows what probability of error would
  be achieved by uncoded transmission (repetition coding) if the
  transmitter is given extra power beyond that predicted by Shannon
  capacity. This extra power corresponds to that required to run one
  iteration of the decoder. Once $\gamma$ gets large, there is
  effectively no point in doing coding.}
\label{fig:uncoded}
\end{center}
\end{figure}


\subsection{Upper bounds on complexity}
\label{sec:upbound}

It is unclear how tight the lower bounds given earlier in this section
are. The most shocking aspect of the lower bounds is that they predict
a double exponential improvement in probability of error with the
number of iterations. This is what is leading to the potential for
weakly capacity-achieving codes. To see the order-optimality of the
bound in principle, we will ``cheat'' and exploit the fact that our
model for iterative decoding in Section~\ref{sec:sysmod} does not
limit either the size of the messages or the computational power of
each node in the decoder. This allows us to give upper bounds on the
number of iterations required for a given performance.
\begin{theorem}
\label{thm:upperbd}
There exists a code of rate $R<C$ such that the required neighborhood
size to achieve $\pe$ average probability of error is upper bounded by  
\begin{equation}
n\leq \frac{\lo{\frac{1}{\pe}}}{E_r(R)} 
\end{equation}
where $E_r(R)$ is the random-coding error exponent for the
channel \cite{Gallager}. The required number of iterations to achieve
this neighborhood size is bounded above by 
\begin{equation} \label{eqn:looseiterationbound}
l - 2 \leq 2 \frac{\lo{n}}{\lo{\alpha}}.
\end{equation}
\end{theorem}
\begin{proof}
This ``code'' is basically an abuse of the definitions. We simply use
a rate-$R$ random code of length $n$ from \cite{Gallager} where each
code symbol is drawn iid. Such random codes if decoded using ML
decoding satisfy
\begin{equation} \label{eq:pebit}
\pe_P\leq \pe_{\text{block}}\leq \exp(-nE_r(R)).
\end{equation}

The decoder for each bit needs at most $n$ channel-output symbols to
decode the block (and hence any particular bit).

Now it is enough to show an upper bound on the number of iterations,
$l$. Consider a regular tree structure imposed on the code with a
branching factor of $\alpha$ and thus overall degree $\alpha+1$. Since
the tree would have $\alpha^d$ nodes in it at depth $d$, a required
depth of $d = \frac{\lo{n}}{\lo{\alpha}} + 1$ is sufficient to
guarantee that everything within a block is connected.

Designate some subset of computational nodes as responsible for
decoding the individual message bits. At each iteration, the
``message'' transmitted by a node is just the complete list of its own
observation plus all the messages that node has received so
far. Because the diameter of a tree is no more than twice its depth,
at the end of $2d$ iterations, all the nodes will have received all
the values of received symbols in the neighborhood. They can then each
ML decode the whole block, with average error probability given
by~\eqref{eq:pebit}. The result follows. 
\end{proof}

For both the AWGN channel and BSC, this bound recovers the basic
behavior that is needed to have the probability of error drop
doubly-exponentially in the number of iterations. For the BSC, it is
also clear that since $E_r(R) = D(C^{-1}(R)||p)$ for rates $R$ in the
neighborhood of capacity, the upper and lower bounds essentially agree
on the asymptotic neighborhood size when $\pe \rightarrow 0$. The only
difference comes in the number of iterations. This is at most a factor
of $2$ and so has the same effect as a slightly different $\xi_D$ in
terms of the shape of the curves and optimizing transmit power.

We note here that this upper bound points to the fact that the
decoding model of Section~\ref{sec:sysmod} is too powerful rather than
being overly constraining. It allows free computations at each node
and unboundedly large messages. This suggests that the lower bounds
are relevant, but it is unclear whether they are actually attainable
with any implementable code. We delve further into this in
Section~\ref{sec:conclusions}.


\section{The gap to capacity and related work} \label{sec:gapandrelated}

Looking back at our bounds of Theorems \ref{thm:basicBSCbound} and
\ref{thm:basicAWGNbound}, they seem to suggest that a certain minimum
number ($\log_{\alpha} f(R,P_T)$) of iterations are required and after
that, the probability of error can drop doubly exponentially with
additional iterations. This parallels the result of \cite[Theorem
  5]{Lentmaier05} for regular LDPCs that essentially implies that
regular LDPCs can be considered weakly {\em certainty-achieving}
codes. However, our bounds above indicate that iterative decoding
might be compatible with weakly {\em capacity-achieving} codes as
well. Thus, it is interesting to ask how the complexity behaves if we
operate very close to capacity. Following tradition, denote the
difference between the channel capacity $C(P)$ and the rate $R$ as the
$gap=C(P)-R$.

Since our bounds are general, it is interesting to compare them with
the existing specialized bounds in the vicinity of capacity. After
first reviewing a trivial bound in Section~\ref{sec:baselinebound}  to
establish a baseline, we review some key results in the literature in
Section~\ref{sec:relatedwork}. Before we can give our results, we take
another look at the waterfall curve in Figure~\ref{fig:waterfall1} and
notice that there are a number of ways to approach the Shannon
limit. We discuss our approach in Section~\ref{sec:gap} before giving
our lower bounds to the number of iterations in
Section~\ref{sec:capacityvicinity}.

\subsection{The trivial bound for the BSC} \label{sec:baselinebound}

Given a crossover probability $p$, it is important to note that there
exists a semi-trivial bound on the neighborhood size that only depends
on the $\pe$. Since there is at least one configuration of the
neighborhood that will decode to an incorrect value for this bit, it
is clear that 
\begin{equation} \label{eqn:trivialbound}
\pe \geq p^n.
\end{equation}
This implies that the number of computational iterations for a code
with maximum decoding degree $\alpha+1$ is lower bounded by $\frac{\log
  \log \frac{1}{\pe} - \log \log \frac{1}{p}}{\log \alpha}$. This
trivial bound does not have any dependence on the capacity and so does
not capture the fact that the complexity should increase inversely as
a function of $gap$ as well.

\subsection{Prior work} \label{sec:relatedwork}
There is a large literature relating to codes that are specified by
sparse graphs. The asymptotic behavior as these codes attempt to
approach Shannon capacity is a central question in that
literature. For regular LDPC codes, a result in Gallager's
Ph.D.~thesis~\cite[Pg. 40]{gallagerThesis} shows that the average
degree of the graph (and hence the average number of operations per
iteration) must diverge to infinity in order for these codes to
approach capacity even under ML decoding. It turns out that it is not
hard to specialize our Theorem~\ref{thm:basicBSCbound} to regular LDPC
codes and have it become tighter along the way. Such a modified bound
would show that as the gap \textit{from Gallager's rate bound}
converges to zero, the number of iterations must diverge to
infinity. However, it would permit double-exponential improvements in
the probability of error as the number of iterations increased.

More recently, in~\cite[Pg. 69]{khandekarthesis} and \cite{mceliece},
Khandekar and McEliece conjectured that for all sparse-graph codes,
the number of iterations must scale either multiplicatively as 
\begin{equation} \label{eqn:mceliece1}
\Omega\left(\lo{\frac{1}{\pe}}\frac{1}{gap}\right),
\end{equation}
or additively as
\begin{equation} \label{eqn:mceliece2}
\Omega\left(\frac{1}{gap}+\lo{\frac{1}{\pe}}\right)
\end{equation}
in the near neighborhood of capacity. Here we use the $\Omega$
notation to denote lower-bounds in the order sense of
\cite{CLR89}. This conjecture is based on a graphical argument for the
message-passing decoding of sparse-graph codes over the BEC. The
intuition was that the bound should also hold for general memoryless
channels, since the BEC is the channel with the simplest decoding.

Recently, the authors in~\cite{sason} were able to formalize and prove
a part of the Khandekar-McEliece conjecture for three important
families of sparse-graph codes, namely the LDPC codes, the
Accumulate-Repeat-Accumulate (ARA) codes, and the Irregular-Repeat
Accumulate (IRA) codes. Using some remarkably simple bounds, the
authors demonstrate that the number of iterations usually scales as
$\Omega(\frac{1}{gap})$ for Binary Erasure Channels (BECs). If,
however, the fraction of degree-$2$ nodes for these codes converges to
zero, then the bounds in~\cite{sason} become trivial. The authors note
that all the known traditionally capacity-achieving sequences of these
code families have a non-zero fraction of degree-$2$ nodes. 

In addition, the bounds in \cite{sason} do not imply that the number
of decoding iterations must go to infinity as $\pe \rightarrow 0$. So
the conjecture is not yet fully resolved. We can observe however that
both of the conjectured bounds on the number of decoding iterations
have only a singly-exponential dependence of the probability of error
on the number of iterations. The multiplicative bound
\eqref{eqn:mceliece2} behaves like a block or convolutional code with
an error-exponent of $K\times gap$ and so, by the arguments of
Section~\ref{sec:firstmagic}, is not compatible with such codes being
weakly capacity achieving in our sense. However, it turns out that the
additive bound \eqref{eqn:mceliece1} {\em is} compatible with being
weakly capacity achieving. This is because the main role of the
double-exponential in our derivation is to allow a second logarithm to
be taken that decoupled the term depending on the transmit power from
the one that depends on the probability of error. The conjectured
additive bound \eqref{eqn:mceliece1} has that form already.

\subsection{`Gap' to capacity}
\label{sec:gap}

In the vicinity of capacity, the complication is that for any finite
probability of bit-error, it is in principle possible to communicate
at rates {\bf above} the channel capacity. Before transmission, the
$k$ bits could be lossily compressed using a source code to $\approx
(1-h_b(\pe))k$ bits. The channel code could then be used to protect
these bits, and the resulting codeword transmitted over the
channel. After decoding the channel code, the receiver could in
principle use the source decoder to recover the message bits with an
acceptable average probability of bit error. Therefore, for fixed
$\pe$, the maximum achievable rate is $\frac{C}{1-h_b(\pe)}$.

Consequently, the appropriate \textit{total gap} is
$\frac{C}{1-h_b(\pe)}-R$, which can be broken down as sum of two
`gap's
\begin{equation}
\label{eq:gaps}
\frac{C}{1-h_b(\pe)}-R=\left\{\frac{C}{1-h_b(\pe)}-C\right\}+\{C-R\}
\end{equation}
The first term goes to zero as $\pe\rightarrow 0$ and the second term
is the intuitive idea of gap to capacity. 

The traditional approach of error exponents is to study the behavior
as the gap is fixed and $\pe \rightarrow 0$. Considering the
error exponent as a function of the gap reveals something about how
difficult it is to approach capacity. However, as we have seen in the
previous section, our bounds predict double-exponential improvements
in the probability of error with the number of iterations. In that
way, our bounds share a qualitative feature with the trivial bound of
Section~\ref{sec:baselinebound}. 

It turns out that the bounds of Theorems~\ref{thm:basicBSCbound} and
\ref{thm:basicAWGNbound} do not give very interesting results if we
fix $\pe > 0$ and let $R \rightarrow C$. We need $\pe \rightarrow 0$
alongside $R \rightarrow C$. To capture the intuitive idea of gap,
which is just the second term in \eqref{eq:gaps}, we want to be able
to assume that the effect of the second term dominates the first. This
way, we can argue that the decoding complexity increases to infinity
as $gap\rightarrow 0$ and not just because $\pe \rightarrow 0$. For
this, it suffices to consider $\pe=gap^\beta$ for $\beta > 1$. Our
proof actually gives a result for $\pe=gap^\beta$ for any $\beta>0$.

\subsection{Lower bound on iterations for regular decoding in the
  vicinity of capacity} \label{sec:capacityvicinity}

Theorems~\ref{thm:basicBSCbound} and \ref{thm:basicAWGNbound} can be
expanded asymptotically in the vicinity of capacity to see the order
scaling of the required neighborhood size with the gap to
capacity. Essentially, this shows that the neighborhood size must grow
at least proportional to $\frac{1}{gap^2}$ unless the average
probability of bit error is dropping so slowly with $gap$ that the
dominant gap is actually the $\left(\frac{C}{1-h_b(\pe)}-C\right)$
term in \eqref{eq:gaps}.

\begin{theorem} \label{thm:lbn}
For the problem as stated in Section~\ref{sec:sysmod}, we obtain the
following lower bounds on the required neighborhood size $n$ for $\pe
= gap^{\beta}$ and $gap \rightarrow 0$.  

For the BSC,
\begin{itemize}
\item For $\beta<1$, $n= \Omega\left(\frac{\lo{1/gap}}{gap^{2\beta}}\right)$.
\item For $\beta\geq 1$, $n= \Omega\left(\frac{\lo{1/gap}}{gap^2}\right)$.
\end{itemize}
For the AWGN channel,
\begin{itemize}
\item For $\beta<1$, $n= \Omega\left(\frac{1}{gap^{2\beta}}\right)$.
\item For $\beta\geq 1$, $n= \Omega\left(\frac{1}{gap^2}\right)$.
\end{itemize}
\end{theorem}
\begin{proof}
We give the proof here in the case of the BSC with some details
relegated to the Appendix. The AWGN case follows analogously, with
some small modifications that are detailed in
Appendix~\ref{app:approxawgn}. 

Let the code for the given BSC $P$ have rate $R$. Consider BSC
channels $G$, chosen so that $C(G)<R<C(P)$, where $C(\cdot{})$ maps a
BSC to its capacity in bits per channel use. Taking $\lo{\cdot{}}$ on
both sides of~\eqref{eq:peip} (for a fixed $g$),
\begin{equation}
\label{eq:afterlog}
\lo{\pe_P} \geq  \lo{h_b^{-1}\left(\delta(G)\right)}-1-nD\left(g||p\right)-\epsilon \sqrt{n}\log_2\left(\frac{g(1-p)}{p(1-g)}\right).
\end{equation}
Rewriting~\eqref{eq:afterlog},
\begin{equation}
\label{eq:afterlog2}
nD\left(g||p\right)+\epsilon \sqrt{n}\log_2\left(\frac{g(1-p)}{p(1-g)}\right)+\lo{\pe_P} - \lo{h_b^{-1}\left(\delta(G)\right)}+1\geq 0.
\end{equation}
This equation is quadratic in $\sqrt{n}$. The LHS potentially has two
roots. If both the roots are not real, then the expression is always
positive, and we get a trivial lower bound of $\sqrt{n}\geq
0$. Therefore, the cases of interest are when the two roots are
real. The larger of the two roots is a lower bound on $\sqrt{n}$.  

Denoting the coefficient of $n$ by $a = D\left(g||p\right)$, that of
$\sqrt{n}$ by $b = \epsilon\log_2\left(\frac{g(1-p)}{p(1-g)}\right)$, and
the constant terms by $c = \lo{\pe_P} -
\lo{h_b^{-1}\left(\delta(G)\right)}+1$ in~\eqref{eq:afterlog2}, the
quadratic formula then reveals
\begin{equation}
\label{eq:quad}
\sqrt{n}\geq \frac{-b+\sqrt{b^2-4ac}}{2a}.
\end{equation}
Since the lower bound holds for all $g$ satisfying $C(G)<R=C-gap$, we
substitute $g^*=p+gap^r$, for some $r<1$ and small $gap$. This choice
is motivated by examining Figure~\ref{fig:gopt}. The constraint $r<1$
is imposed because it ensures $C(g^*)<R$ for small enough $gap$.
\begin{lemma}
\label{lem:rlessthan1}
In the limit of $gap\rightarrow 0$, for $g^*=p+gap^r$ to satisfy $C(g^*)<R$, it suffices that $r$ be less than 1.
\end{lemma}
\begin{proof}
\begin{eqnarray*}
C(g^*) &=& C(p+gap^r)\\
 &=& C(p) + gap^r\times C'(p) + o(gap^r)\\
 &\leq & C(p) - gap=R,
\end{eqnarray*}
for small enough $gap$ and $r<1$. The final inequality holds since
$C(p)$ is a monotonically-decreasing concave-$\cap$ function for a BSC
with $p < \frac{1}{2}$ whereas $gap^r$ increases faster than any
linear function of $gap$ when $gap$ is small enough.
\end{proof}

\vspace{0.1in}

\begin{figure}[htb]
\begin{center}
\includegraphics[scale=0.7]{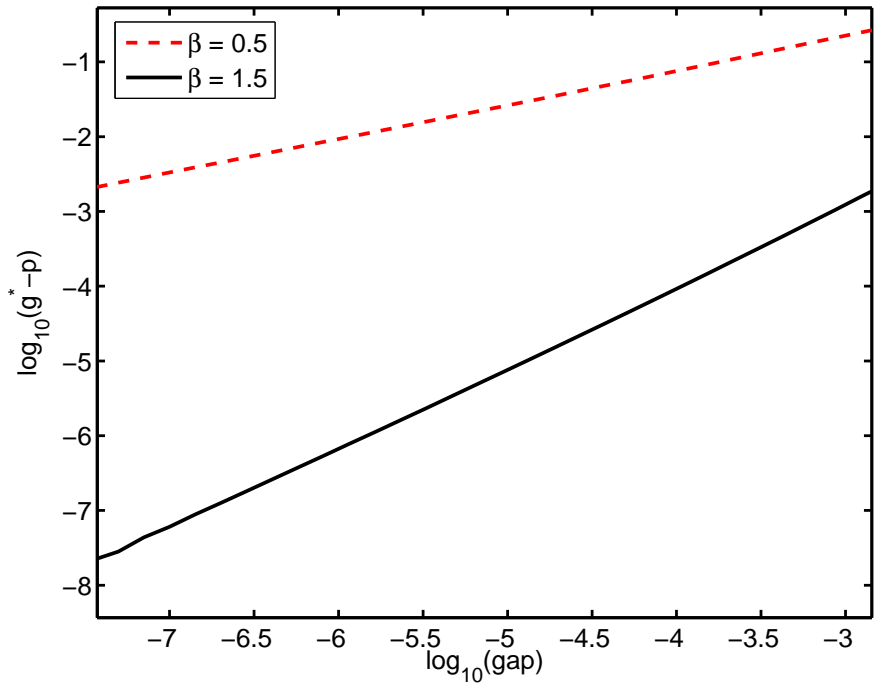}
\caption{The behavior of $g^*$, the optimizing value of $g$ for the
  bound in Theorem~\ref{thm:basicBSCbound}, with $gap$. We plot
  $\log(g_{\mbox{opt}}-p)$ vs $\log(gap)$. The resulting straight
  lines inspired the substitution of $g^*=p+gap^r$.}  
\label{fig:gopt}
\end{center}
\end{figure}

In steps, we now Taylor-expand the terms on the LHS of
\eqref{eq:afterlog2} about $g=p$. 


\begin{lemma}[Bounds on $h_b(p)$ and $h_b^{-1}(p)$ from \cite{HariPersonalComOct}]
\label{lem:hbbound}
For all $d > 1$, and for all $x\in [0,\frac{1}{2}]$ and $y\in [0,1]$
\begin{eqnarray}
h_b(x) &\geq&  2 x \label{eqn:binentropylowerbound}\\
h_b(x) &\leq&  2 x^{1-1/d} d/\ln(2) \\
h_b^{-1}(y) &\geq& y^{\frac{d}{d-1}}\left(\frac{\ln(2)}{2d}
\right)^{\frac{d}{d-1}} \label{eqn:invbinentropylowerbound} \\
h_b^{-1}(y) &\leq& \frac{1}{2} y. \label{eqn:invbinentropyupperbound}
\end{eqnarray}
\end{lemma}
\begin{proof}
See Appendix \ref{app:lemhbbound}.
\end{proof}
\vspace{0.1in}

\begin{lemma}
\label{lem:approxloghbinvdeltag}
\begin{equation} \label{eqn:approxloghbinvdeltag}
\frac{d}{d-1}r\;\lo{gap} -1 + K_1  + o(1)
\leq
\lo{h_b^{-1}\left(\delta(g^*)\right)}
\leq
r\;\lo{gap} -1 + K_2 + o(1) 
\end{equation}
where $K_1= \frac{d}{d-1}\left(\lo{\frac{h_b'(p)}{C(p)}} +
\lo{\frac{\ln(2)}{d}}\right)$ where $d > 1$ is arbitrary and 
$K_2 = \log_2(\frac{h_b'(p)}{C(p)})$. 
\end{lemma}
\begin{proof}
See Appendix~\ref{app:lemapproxloghbinvdeltag}.
\end{proof}
\vspace{0.1in}

\begin{lemma}
\label{lem:approxdiv}
\begin{equation}
D(g^*||p) = \frac{gap^ {2r}}{2 p(1-p)\ln(2)}(1 + o(1)).
\end{equation}
\end{lemma}
\begin{proof}
See Appendix~\ref{app:lemapproxdiv}.
\end{proof}
\vspace{0.1in}

\begin{lemma}
\label{lem:approxlogfrac}
\begin{eqnarray*}
\log_2\left(\frac{g^*(1-p)}{p(1-g^*)}\right)=
\frac{gap^r}{p(1-p)\ln(2)}(1 + o(1)).
\end{eqnarray*}
\end{lemma}
\begin{proof}
See Appendix~\ref{app:lemapproxlogfrac}.
\end{proof}
\vspace{0.1in}

\begin{lemma}
\label{lem:approxepsilon}
\begin{eqnarray*}
\sqrt{\frac{r }{K(p)}} \sqrt{\lo{\frac{1}{gap}}}
(1 + o(1))
\leq \epsilon \leq 
\sqrt{\frac{r d }{(d-1)K(p)}} \sqrt{\lo{\frac{1}{gap}}}
(1 + o(1))
\end{eqnarray*}
where $K(p)$ is from \eqref{eqn:Kdef}. 
\end{lemma}
\begin{proof}
See Appendix~\ref{app:lemapproxepsilon}.
\end{proof}
\vspace{0.1in}

If $c< 0$, then the bound \eqref{eq:quad} is guaranteed to be
positive. For $\pe_P = gap^\beta$, the condition $c<0$ is equivalent
to 
\begin{equation}
\label{eq:clesszero}
\beta\lo{gap}-\lo{h_b^{-1}\left(\delta(g^*)\right)}+1<0
\end{equation}
Since we want~\eqref{eq:clesszero} to be satisfied for all small
enough values of $gap$, we can use the approximations in
Lemma~\ref{lem:approxloghbinvdeltag}--\ref{lem:approxepsilon} and ignore
constants to immediately arrive at the following sufficient condition 
\begin{eqnarray*}
\beta\lo{gap}-\frac{d}{d-1}r\lo{gap}& <& 0\\
\text{i.e. }\;r&<&\frac{\beta (d-1)}{d},
\end{eqnarray*}
where $d$ can be made arbitrarily large. 
Now, using the approximations in
Lemma~\ref{lem:approxloghbinvdeltag} and Lemma~\ref{lem:approxlogfrac}, and
substituting them into~\eqref{eq:quad}, we can evaluate the solution
of the quadratic equation. 

As shown in Appendix~\ref{app:nbsc}, this gives us the following lower
bound on $n$.
\begin{equation}
n \geq \Omega\left(\frac{\lo{1/gap}}{gap^{2r}}\right)
\end{equation}
for any $r<\min\{\beta,1\}$. Theorem~\ref{thm:lbn} follows.
\end{proof}
\vspace{0.1in}

The lower bound on neighborhood size $n$ can immediately be converted
into a lower bound on the minimum number of computational iterations
by just taking $\log_\alpha (\cdot{})$. Note that this is not a
comment about the degree of a potential sparse graph that defines the
code. This is just about the maximum degree of the decoder's
computational nodes and is a bound on the number of computational
iterations required to hit the desired average probability of error.

It turns out to be easy to show that the upper bound of
Theorem~\ref{thm:upperbd} gives rise to the same $\frac{1}{gap^2}$
scaling on the neighborhood size. This is because the random-coding
error exponent in the vicinity of the capacity agrees with the
sphere-packing error exponent which just has the quadratic term coming
from the KL divergence. However, when we translate it from
neighborhoods to iterations, the two bounds asymptotically differ by a
factor of $2$ that comes from \eqref{eqn:looseiterationbound}.

The lower bounds are plotted in Figure~\ref{fig:variousbeta} for
various different values of $\beta$ and reveal a $\log \frac{1}{gap}$
scaling to the required number of iterations when the decoder has
bounded degree for message passing. This is much larger than the
trivial lower bound of $\log \log \frac{1}{gap}$ but is much smaller
than the Khandekar-McEliece conjectured $\frac{1}{gap}$ or
$\frac{1}{gap}\lo{\frac{1}{gap}}$ scaling for the number of iterations
required to traverse such paths toward certainty at capacity.

\begin{figure}[htb]
 \begin{center}
 \includegraphics[scale=0.7]{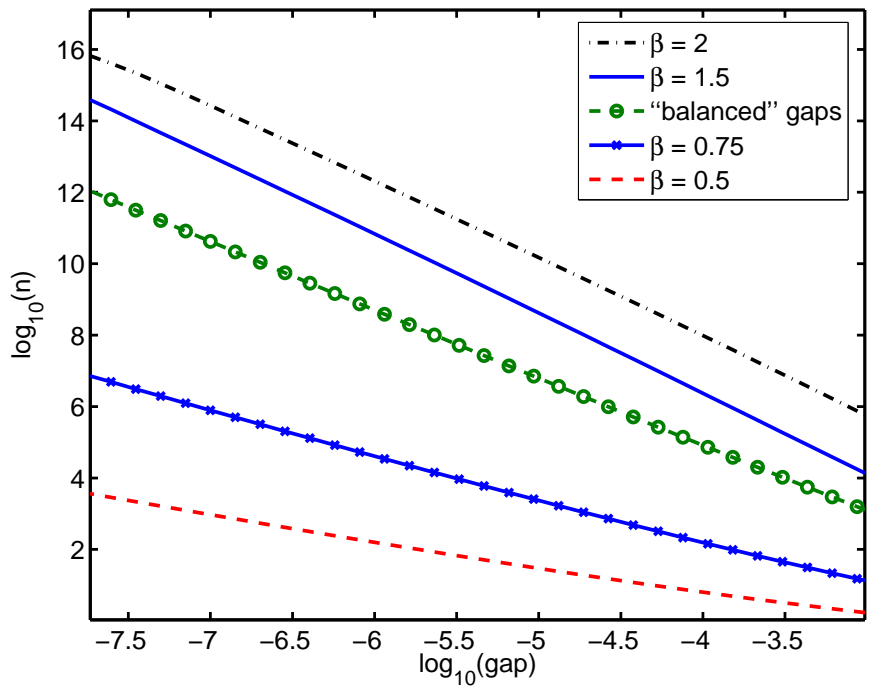}
 \caption{Lower bounds for neighborhood size vs the gap to capacity for
   $\pe=gap^\beta$ for various values of $\beta$. The curve titled
   ``balanced'' gaps shows the behavior for
   $\frac{C}{1-h_b(\pe)}-C=C-R$, that is, the two `gaps' are equal. The
   curves are plotted by brute-force optimization of~\eqref{eq:peip},
   but reveal slopes that are as predicted in Theorem~\ref{thm:lbn}.}
 \label{fig:variousbeta}
 \end{center}
\end{figure}

\section{Conclusions and future work}
\label{sec:conclusions}

In this work, we use the inherently local nature of message-passing
decoding algorithms to derive lower bounds on the number of
iterations. It is interesting to note that with so few assumptions on
the decoding algorithm and the code structure, the number of
iterations still diverges to infinity as $gap\rightarrow 0$. As
compared to~\cite{grover} where a similar approach is adopted, the
bounds here are stronger, and indeed tight in an order-sense for the
decoding model considered. To show the tightness (in order) of these
bounds, we derived corresponding upper bounds that behave similar to
the lower bounds, but these exploit a loophole in our complexity
model. Our model only considers the limitations induced by the
internal communication structure of the decoder --- it does not
restrict the computational power of the nodes within the decoder. Even
so, there is still a significant gap between our upper and lower
bounds in terms of the constants and we suspect this is largely
related to the known looseness of the sphere-packing bound
\cite{sasonshamai}, as well as our coarse bounding of the required
graph diameter. Our model also does not address the power requirements
of encoding.

Because we assume little about the code structure, the bounds here are
much more optimistic than those in~\cite{sason}. However, it is
unclear to what extent the optimism of our bound is an artifact. After
all, \cite{Lentmaier05} does get double-exponential reductions in
probability of error with additional iterations, but for a family of
codes that does not seem to approach capacity. This suggests that an
investigation into expander codes might help resolve this question
since expander codes can approach capacity, be decoded using a circuit
of logarithmic depth (like our iterations), and achieve error
exponents with respect to the overall block length
\cite{BargZemov}. It may very well be that expanders or expander-like
codes can be shown to be weakly capacity achieving in our sense.

For any kind of capacity-achieving code, we conjecture that the
optimizing transmit power will be the sum of three terms
$$P_T^* = C^{-1}(R) + \mbox{Tech}(\vec{\xi}, \alpha, E_{node}, R) \pm
A(\pe, R,\vec{\xi}, \alpha, E_{node}).$$
\begin{itemize} 
\item $C^{-1}(R)$ is the prediction from Shannon's capacity.
\item $\mbox{Tech}(\vec{\xi}, \alpha, E_{node}, R)$ is the minimum
  extra transmit power that needs to be used asymptotically to help
  reduce the difficulty of encoding and decoding for the given
  application and implementation technology. Solving
  \eqref{eqn:poweroffsetequation} and subtracting $C^{-1}(R)$ gives a
  heuristic target value to aim for, but it remains an open problem to
  get a tight estimate for this term.
\item $A(\pe, R,\vec{\xi}, \alpha, E_{node})$ is an amount by which we
  should increase or reduce the transmit power because we are willing
  to tolerate some finite probability of error and the non-asymptotic
  behavior is still significant. This term should go to zero as $\pe
  \rightarrow 0$.
\end{itemize}

Understanding the second term $\mbox{Tech}(\vec{\xi}, \alpha,
E_{node}, R)$ above is what is needed to give principled answers
regarding how close to capacity should the transmitter operate.

The results here indicate that strongly capacity-achieving coding
systems are not possible if we use the given model of iterative
decoding. There are a few possibilities worth exploring.
\begin{enumerate}
\item Our model of iterative decoding left out some real-world
   computational capability that could be exploited to dramatically
   reduce the required power consumption. There are three natural
   candidates here. 
\begin{itemize}
   \item {\em Selective and adaptive sleep:} In the current model,
     all computational nodes are actively consuming power for all the
     iterations. If there was a way for computational nodes to
     adaptively turn themselves off and use {\bf no power} while
     sleeping, then the results might change. We suspect that bounding
     the performance of such systems will require some sort of
     neighborhood-oriented analogies to the bounds for variable-block-length
     coding \cite{ForneyErasure, Burnashev}.

  \item {\em Dynamically reconfigurable circuits:} In the current
    model, the connectivity structure of computational nodes is
    fixed and considered as unchanging wiring. If there was a way for
    computational nodes to dynamically rewire who their neighbors
    are (for example by moving themselves in the combined spirit of
    \cite{TseMobility, RoseCosmic, manishAllerton07}), this might
    change the results. 

  \item {\em Feedback:} In \cite{OurUpperboundPaper}, a general scheme
    is presented that achieves an infinite computational error
    exponent by exploiting noiseless channel-output feedback as well
    as an infinite amount of common randomness. If such a scheme could
    be implemented, it would presumably be strongly capacity achieving
    as both the transmission and processing power could remain finite
    while having arbitrarily low average probability of bit
    error. However, we are unaware if either this scheme or any of the
    encoding strategies that claim to deliver ``linear-time'' encoding
    and decoding with an error exponent (e.g.~\cite{BargZemov,
      GuruswamiIndyk}) are actually implementable in a way that uses
    finite total power.

\end{itemize}
\item Strong or even weakly capacity-achieving communication systems
  may be possible using infallible computational entities but may be
  impossible to achieve using unreliable computational nodes that must
  burn more power (i.e.~raise the voltages) to be more reliable
  \cite{SomeLavPaper}. 

\item Either strongly or weakly capacity-achieving communication
  systems might be impossible on thermodynamic grounds. Decoding in
  some abstract sense is related to the idea of cooling a part of a
  system \cite{RujanTemperature}. Since an implementation can be
  considered a collection of Maxwell Daemons, this might be useful
  to rule out certain models of computation as being aphysical.
\end{enumerate}

Finally, the approach here should be interesting if extended to a
multiuser context where the prospect of causing interference makes it
less easy to improve reliability by just increasing the transmit
power. There, it might give some interesting answers as to what kind
of computational efficiency is needed to make it asymptotically worth
using multiterminal coding theory.

\appendices

\section{Proof of Theorem~\ref{thm:basicBSCbound}: lower bound on $\pe_P$ for the BSC}
\label{app:peip}

The idea of the proof is to first show that the average probability of
error for any code must be significant if the channel were a much
worse BSC. Then, a mapping is given that maps the probability of an
individual error event under the worse channel to a lower-bound on its
probability under the true channel. This mapping is shown to be
convex-$\cup$ in the probability of error and this allows us to use
this same mapping to get a lower-bound to the average probability of
error under the true channel. We proceed in steps, with the lemmas
proved after the main argument is complete.

\begin{proof} 
Suppose we ran the given encoder and decoder over a test channel $G$
instead. 
\begin{lemma}[Lower bound on $\pe$ under test channel $G$.] 
\label{lem:lbg} If a rate-$R$ code is used over a channel $G$ with
$C(G) < R$, then the average probability of bit error satisfies
\begin{equation}
\label{eq:peig}
\pe_G \geq h_b^{-1}\left(\delta(G)\right)
\end{equation}
where $\delta(G) = 1 - \frac{C(G)}{R}$. This holds for any channel
model $G$, not just BSCs.
\end{lemma}
\begin{proof}
See Appendix~\ref{app:peig}.
\end{proof}
\vspace{0.1in}

Let  $\mk{b}$ denote the entire message, and let $\m{x}$ be the
corresponding codeword. Let the common randomness available to the
encoder and decoder be denoted by the random variable $U$, and its
realizations by $u$. 

Consider the $i$-th message bit $B_i$. Its decoding is performed by
observing a particular decoding neighborhood\footnote{For any given
  decoder implementation, the size of the decoding neighborhood might
  be different for different bits $i$. However, to avoid unnecessary
  complex notation, we assume that the neighborhoods are all the same
  size $n$ corresponding to the largest possible neighborhood
  size. This can be assumed without loss of generality since smaller
  decoding neighborhoods can be supplemented with additional channel
  outputs that are ignored by the decoder.} of channel outputs
$\nbd{y}$. The corresponding channel inputs are denoted by $\nbd{x}$,
and the relevant channel noise by $\nbd{z}=\nbd{x}\oplus \nbd{y}$
where $\oplus$ is used to denote modulo $2$ addition. The decoder just
checks whether the observed $\nbd{y} \in \mathcal{D}_{y,i}(0, u)$ to
decode to $\widehat{B}_i = 0$ or whether $\nbd{y} \in
\mathcal{D}_{y,i}(1,u)$ to decode to $\widehat{B}_i = 1$.

For given $\nbd{x}$, the error event is equivalent to $\nbd{z}$
falling in a decoding region $\mathcal{D}_{z,i}(\nbd{x}, \mk{b}, u) =
\mathcal{D}_{y,i}(1 \oplus b_i, u) \oplus \nbd{x}$. Thus by the linearity
of expectations, \eqref{eq:peig} can be rewritten as:
\begin{equation} \label{eqn:initialconstraint}
\frac{1}{k} \sum_i \frac{1}{2^{k}} \sum_{\mk{b}} \sum_u \Pr(U = u) 
\Pr_G(\nbd{Z} \in \mathcal{D}_{z,i}(\nbd{x}(\mk{b},u), \mk{b}, u))
\geq h_b^{-1}\left(\delta(G)\right).
\end{equation}

The following lemma gives a lower bound to the probability of an event
under channel $P$ given a lower bound to its probability under channel
$G$. 

\begin{lemma} \label{lem:bscmapping}
Let $A$ be a set of BSC channel-noise realizations $\mn{z}$ such that
$\Pr_G(A) = \delta$. Then 
\begin{equation}
\Pr_P(A) \geq f\left(\delta \right)
\end{equation}
where
\begin{equation}
f(x) = \frac{x}{2}\;2^{-nD(g||p)}\left(\frac{p(1-g)}{g(1-p)}\right)^{\epsilon(x) \sqrt{n}}
\end{equation}
is a convex-$\cup$ increasing function of $x$ and 
\begin{equation}
\epsilon(x) = \sqrt{\frac{1}{K(g)}\lo{\frac{2}{x}}}.
\end{equation} 
\end{lemma}
\begin{proof} See Appendix~\ref{app:convexitybsc}.
\end{proof}
\vspace{0.1in}
Applying Lemma~\ref{lem:bscmapping} in the style of \eqref{eqn:initialconstraint}
tells us that:

\begin{eqnarray}
\pe_P 
 & = &
\frac{1}{k} \sum_i \frac{1}{2^{k}} \sum_{\mk{b}} \sum_u \Pr(U = u) 
\Pr_P\left(\nbd{Z} \in \mathcal{D}_{z,i}(\nbd{x}(\mk{b},u), \mk{b}, u)\right) \nonumber \\
 & \geq &
\frac{1}{k} \sum_i \frac{1}{2^{k}} \sum_{\mk{b}} \sum_u \Pr(U = u) 
f(\Pr_G\left(\nbd{Z} \in \mathcal{D}_{z,i}(\nbd{x}(\mk{b},u), \mk{b}, u)\right)). \label{eqn:almostbscobjective}
\end{eqnarray}

But the increasing function $f(\cdot)$ is also convex-$\cup$ and thus
\eqref{eqn:almostbscobjective} and \eqref{eqn:initialconstraint} imply
that 
\begin{eqnarray*}
\pe_P 
 & \geq &
f(\frac{1}{k} \sum_i \frac{1}{2^{k}} \sum_{\mk{b}} \sum_u \Pr(U = u) 
\Pr_G\left(\nbd{Z} \in \mathcal{D}_{z,i}(\nbd{x}(\mk{b},u), \mk{b},
u)\right)) \\
 & \geq & f(h_b^{-1}\left(\delta(G)\right)).
\end{eqnarray*}

This proves Theorem~\ref{thm:basicBSCbound}. 
\end{proof}
\vspace{0.1in} 

At the cost of slightly more complicated notation, by following the
techniques in \cite{OurUpperboundPaper}, similar results can be proved
for decoding across any discrete memoryless channel by using
Hoeffding's inequality in place of the Chernoff bounds used here in
the proof of Lemma~\ref{lem:lbg}. In place of the KL-divergence term
$D(g||p)$, for a general DMC the arguments give rise to a term $\max_x
D(G_x||P_x)$ that picks out the channel input letter that maximizes
the divergence between the two channels' outputs. For output-symmetric
channels, the combination of these terms and the outer maximization
over channels $G$ with capacity less than $R$ will mean that the
divergence term will behave like the standard sphere-packing bound
when $n$ is large. When the channel is not output symmetric (in the
sense of \cite{Gallager}), the resulting divergence term will behave
like the Haroutunian bound for fixed block-length coding over DMCs
with feedback \cite{Haroutunian}.

\subsection{Proof of Lemma~\ref{lem:lbg}: A lower bound on $\pe_G$.}
\label{app:peig}
\begin{proof}
\begin{eqnarray*}
H(\mk{B})-H(\mk{B}|\m{Y}) &=& I(\mk{B};\m{Y}) \leq  I(\m{X};\m{Y}) \leq m C(G).
\end{eqnarray*}
Since the $Ber(\frac{1}{2})$ message bits are iid, $H(\mk{B})=k$. Therefore,
\begin{equation}
\frac{1}{k}H(\mk{B}|\m{Y})\geq 1- \frac{C(G)}{R}.
\end{equation}
Suppose the message bit sequence was decoded to be
$\mk{\widehat{B}}$. Denote the error sequence by
$\mk{\widetilde{B}}$. Then,
\begin{equation}
\mk{B}=\mk{\widetilde{B}} \oplus \mk{\widehat{B}},
\end{equation}
where the addition $\oplus$ is modulo 2. The only complication is the
possible randomization of both the encoder and decoder. However, note
that even with randomization, the true message $\mk{B}$ is independent
of $\mk{\widehat{B}}$ conditioned on $\m{Y}$. Thus,
\begin{eqnarray*}
H(\mk{\widetilde{B}}|\m{Y})&=&H(\mk{\widehat{B}} \oplus \mk{B}|\m{Y})\\
&=& H(\mk{\widehat{B}} \oplus \mk{B}|\m{Y}) + I(\mk{B}; \mk{\widehat{B}}|\m{Y}) \\
&=& H(\mk{\widehat{B}} \oplus \mk{B}|\m{Y}) - H(\mk{B}|\m{Y},\mk{\widehat{B}})+H(\mk{B}|\m{Y})\\
&=& I(\mk{\widehat{B}} \oplus \mk{B};\mk{\widehat{B}}|\m{Y}) +H(\mk{B}|\m{Y})\\
&\geq &H(\mk{B}|\m{Y}) \\
&\geq &k(1- \frac{C(G)}{R}).
\end{eqnarray*}
This implies
\begin{eqnarray}
\label{eq:markov}
\frac{1}{k} \sum_{i=1}^k H(\widetilde{B}_i|\m{Y})&\geq& 1- \frac{C(G)}{R}.
\end{eqnarray}
Since conditioning reduces entropy, $H(\widetilde{B}_i) \geq
H(\widetilde{B}_i|\m{Y})$. Therefore,
\begin{equation}
\label{eq:sum}
\frac{1}{k} \sum_{i=1}^k H(\widetilde{B}_i) \geq 1- \frac{C(G)}{R}.
\end{equation}
Since $\widetilde{B}_i$ are binary random variables,
$H(\widetilde{B}_i)=h_b(\pei_G)$, where $h_b(\cdot{})$ is the binary
entropy function. Since $h_b(\cdot{})$ is a concave-$\cap$ function,
$h_B^{-1}(\cdot)$ is convex-$\cup$ when restricted to output values
from $[0,\frac{1}{2}]$. Thus, \eqref{eq:sum} together with Jensen's
inequality implies the desired result \eqref{eq:peig}.
\end{proof}

\subsection{Proof of Lemma~\ref{lem:bscmapping}: a lower bound on $\pei_P$ as a function of $\pei_G$.}
\label{app:convexitybsc}
\begin{proof}
First, consider a strongly $G-$typical set of $\nbd{z}$, given by 
\begin{equation}
\typical{}= \{\mn{z}\; s.t.\; \sum_{i=1}^n z_{i} - ng \leq \epsilon \sqrt{n}\}.
\end{equation}
In words, $\typical{}$ is the set of noise sequences with weights
smaller than $ng + \epsilon\sqrt{n}$. The probability of an event $A$
can be bounded using
\begin{eqnarray*}
\delta & = & \Pr_G(\mn{Z} \in A) \\
&= &
\Pr_G(\mn{Z} \in A \cap \typical{}) + 
\Pr_G(\mn{Z} \in A \cap \typicalc{}) \\ 
&\leq & 
\Pr_G(\mn{Z} \in A \cap \typical{}) + 
\Pr_G(\mn{Z} \in \typicalc{}).
\end{eqnarray*}
Consequently,
\begin{equation} \label{eqn:basicAbound1}
\Pr_G(\mn{Z} \in A \cap \typical{}) \geq \delta -\Pr_G(\typicalc{}).
\end{equation}

\begin{lemma} \label{lem:chernoffBSC} The probability of the atypical
  set of Bernoulli-$g$ channel noise $\{Z_i\}$ is bounded above by
\begin{eqnarray}
\label{eq:chernoffs}
\Pr_G\left(\frac{ \sum_{i}^n Z_i - ng}{\sqrt{n}} > \epsilon\right) \leq
2^{-K(g) \epsilon^2 }
\end{eqnarray}
where $K(g)=\underset{0<\eta \leq 1-g}{\inf}\frac{D(g+\eta ||g)}{\eta^2}$. 
\end{lemma}
\begin{proof}
See Appendix~\ref{app:ChernoffBSC}.
\end{proof}
\vspace{0.1in}

Choose $\epsilon$ such that
\begin{eqnarray}
\label{eq:lbep}
\nonumber 2^{-K(g)\epsilon^2}& = & \frac{\delta}{2}\\
\text{i.e.}\;\;\epsilon^2 & = &
\frac{1}{K(g)}\log_2\left(\frac{2}{\delta} \right).
\end{eqnarray}

Thus \eqref{eqn:basicAbound1} becomes
\begin{equation} \label{eqn:basicAbound2}
\Pr_G(\mn{Z} \in A \cap \typical{}) \geq \frac{\delta}{2}.
\end{equation}

Let $n_{\mn{z}}$ denote the number of ones in $\mn{z}$. Then, 
\begin{equation}
\Pr_G(\mn{Z} = \mn{z}) =g^{n_{\mn{z}}}(1-g)^{n-n_{\mn{z}}}.
\end{equation} 
This allows us to lower bound the probability of $A$ under channel law
$P$ as follows:
\begin{eqnarray*}
\Pr_P(\mn{Z} \in A) 
&\geq& 
\Pr_P(\mn{Z} \in A \cap \typical{}) \\
&=&
\sum_{\mn{z}\in A \cap \typical{}} \frac{\Pr_P(\mn{z})}{\Pr_G(\mn{z})}\Pr_G(\mn{z})\\
&=&
\sum_{\mn{z}\in A \cap \typical{}}
\frac{p^{n_{\mn{z}}}(1-p)^{n-n_{\mn{z}}}}{g^{n_{\mn{z}}}(1-g)^{n-n_{\mn{z}}}}
\Pr_G(\mn{z}) \\
&\geq & \frac{(1-p)^n}{(1-g)^n} \sum_{\mn{z}\in A \cap \typical{}} \Pr_G(\mn{z})\left(\frac{p(1-g)}{g(1-p)}\right)^{ng+\epsilon \sqrt{n}}\\
&=&
\frac{(1-p)^n}{(1-g)^n}
\left(\frac{p(1-g)}{g(1-p)}\right)^{ng+\epsilon \sqrt{n}}
\Pr_G(A \cap \typical{})\\
& \geq & \frac{\delta}{2} 2^{-nD\left(g||p\right)}\left(\frac{p(1-g)}{g(1-p)}\right)^{\epsilon \sqrt{n}}.
\end{eqnarray*} 
This results in the desired expression:
\begin{equation}
\label{eq:peip3}
f(x) =  \frac{x}{2}\;2^{-nD\left(g||p\right)}\left(\frac{p(1-g)}{g(1-p)}\right)^{\epsilon(x) \sqrt{n}}.
\end{equation}
where $\epsilon(x) = \sqrt{\frac{1}{K(g)}\log_2\left(  \frac{2}{x}
\right)}$. To see the convexity of $f(x)$, it is useful to
apply some substitutions. Let $c_1 = \frac{2^{-nD(g||p)}}{2} > 0$ and let
$\xi = \sqrt{\frac{n}{K(g) \ln 2}} \ln(\frac{p(1-g)}{g(1-p)})$. Notice that
$\xi < 0$ since the term inside the $\ln$ is less than $1$. Then $f(x)
  = c_1 x \exp(\xi \sqrt{\ln 2 - \ln x})$. 

Differentiating $f(x)$ once results in
\begin{equation}
f'(x) =  c_1 \exp\left(\xi \sqrt{\ln(2) + \ln(\frac{1}{x})}\right) (1 + \frac{-\xi}{2 \sqrt{\ln(2) +  \ln(\frac{1}{x})}}).
\end{equation}
By inspection, $f'(x)>0$ for all $0<x<1$ and thus $f(x)$ is a
monotonically increasing function. Differentiating $f(x)$ twice with
respect to $x$ gives
\begin{equation}
\label{eq:doublediff}
f''(x) = -\xi \frac{c_1 \exp\left(\xi \sqrt{\ln(2) + \ln(\frac{1}{x})}\right)}{2 \sqrt{\ln(2) +  \ln(\frac{1}{x})}}
\left(1 + \frac{1}{2 (\ln(2) + \ln(\frac{1}{x}))} 
- \frac{\xi}{2 \sqrt{\ln(2) +  \ln(\frac{1}{x})}} \right).
\end{equation}
Since $\xi<0$, it is evident that all the terms in
\eqref{eq:doublediff} are strictly positive. Therefore, $f(\cdot{})$
is convex-$\cup$. 
\end{proof}

\subsection{Proof of Lemma~\ref{lem:chernoffBSC}: Bernoulli Chernoff bound} \label{app:ChernoffBSC}
\begin{proof}
Recall that $Z_i$ are iid Bernoulli random variables with mean $g \leq 1/2$.
\begin{equation}
\Pr\left(\frac{\sum_i (Z_i - g)}{\sqrt{n}} \geq \epsilon \right) =
\Pr\left(\frac{\sum_i (Z_i - g)}{n} \geq \widetilde{\epsilon}\right)
\end{equation}
where $\epsilon = \sqrt{n} \widetilde{\epsilon}$ and so $n = \epsilon^2/\widetilde{\epsilon}^2$. Therefore,
\begin{equation}
\label{eq:chernoffbinary}
\Pr(\frac{\sum_i (Z_i - g)}{\sqrt{n}} \geq \epsilon) \leq
[((1-g) + g \exp(s))\times \exp(-s(g+\widetilde{\epsilon}))]^n
\;\;\text{for all}\;\; s \geq 0.
\end{equation}
Choose $s$ satisfying 
\begin{equation} \label{eqn:substitutes}
\exp(-s) = \frac{g}{(1-g)} \times  \left(\frac{1}{(g+\widetilde{\epsilon})} - 1\right).
\end{equation}
It is safe to assume that $g+\widetilde{\epsilon} \leq 1$ since otherwise, the
relevant probability is $0$ and any bound will work. Substituting \eqref{eqn:substitutes}
into \eqref{eq:chernoffbinary} gives 
\begin{eqnarray*}
\Pr\left(\frac{\sum_i (Z_i - g)}{\sqrt{n}} \geq \epsilon \right)
&\leq&
2^{-\frac{D(g+\widetilde{\epsilon}||g)}{\widetilde{\epsilon}^2}\epsilon^2}.
\end{eqnarray*}
This bound holds under the constraint
$\frac{\epsilon^2}{\widetilde{\epsilon}^2}=n$. To obtain a bound that holds
uniformly for all $n$, we fix $\epsilon$, and take the supremum over
all the possible $\widetilde{\epsilon}$ values. 
\begin{eqnarray*}
\Pr\left(\frac{\sum_i (Z_i - g)}{\sqrt{n}} \geq \epsilon\right)
&\leq &\sup_{0<\widetilde{\epsilon}\leq 1-g}
\exp(-\ln(2) \frac{D(g+\widetilde{\epsilon}||g)}{\widetilde{\epsilon}^2}\epsilon^2) \\
&\leq &\exp(- \ln(2)\epsilon^2 \inf_{0<\widetilde{\epsilon}\leq 1-g} \frac{D(g+\widetilde{\epsilon}||g)}{\widetilde{\epsilon}^2}),
\end{eqnarray*}
giving us the desired bound.
\end{proof}

\section{Proof of Theorem~\ref{thm:basicAWGNbound}: lower bound on $\pe_P$ for AWGN channels}
\label{app:peiawgn}

The AWGN case can be proved using an argument almost identical to the
BSC case. Once again, the focus is on the channel noise $Z$ in the
decoding neighborhoods \cite{ChengPersonalComNov}. Notice that
Lemma~\ref{lem:lbg} already applies to this channel even if the power
constraint only has to hold on average over all codebooks and
messages. Thus, all that is required is a counterpart to
Lemma~\ref{lem:bscmapping} giving a convex-$\cup$ mapping from the
probability of a set of channel-noise realizations under a Gaussian
channel with noise variance $\sigma_G^2$ back to their probability
under the original channel with noise variance $\sigma_P^2$.

\begin{lemma} \label{lem:awgnmapping}
Let $A$ be a set of Gaussian channel-noise realizations $\mn{z}$ such that
$\Pr_G(A) = \delta$. Then 
\begin{equation}
\Pr_P(A) \geq f\left(\delta \right)
\end{equation}
where
\begin{equation} \label{eqn:GaussianAsymptoticFdef}
f(\delta) = \frac{\delta}{2}\;
\exp(-nD(\sigma_G^2||\sigma_P^2) - \sqrt{n}
(\frac{3}{2}+2\lon{\frac{2}{\delta}}) \left(\frac{\sigma_G^2}{\sigma_P^2}- 1\right)).
\end{equation}
Furthermore, $f(x)$ is a convex-$\cup$ increasing function in $\delta$
for all values of $\sigma_G^2 \geq \sigma_P^2$. 

In addition, the following bound is also convex whenever $\sigma_G^2 >
\sigma_P^2 \mu(n)$ with $\mu(n)$ as defined in \eqref{eqn:mudef}.
\begin{equation} \label{eqn:GaussianNonAsymptoticFdef}
f_L(\delta) = \frac{\delta}{2}\;
\exp(-nD(\sigma_G^2||\sigma_P^2)  
- \frac{1}{2} \phi(n, \delta) \left(\frac{\sigma_G^2}{\sigma_P^2}-1\right))
\end{equation}
where $\phi(n, \delta)$ is as defined in \eqref{eqn:phidef}.
\end{lemma}
\begin{proof}
See Appendix~\ref{app:AWGNmappingProof}.
\end{proof}
\vspace{0.1in}

With Lemma~\ref{lem:awgnmapping} playing the role of
Lemma~\ref{lem:bscmapping}, the proof for
Theorem~\ref{thm:basicAWGNbound} proceeds identically to that of
Theorem~\ref{thm:basicBSCbound}. 

It should be clear that similar arguments can be used to prove similar
results for any additive-noise models for continuous output
communication channels. However, we do not believe that this will
result in the best possible bounds. Instead, even the bounds for the
AWGN case seem suboptimal because we are ignoring the possibility of a
large deviation in the noise that happens to be locally aligned to the
codeword itself.

\subsection{Proof of Lemma~\ref{lem:awgnmapping}: a lower bound on
  $\pei_P$ as a function of $\pei_G$} \label{app:AWGNmappingProof}
\begin{proof}
Consider the length-$n$ set of $G$-typical additive noise given by 
\begin{equation}
\typical{}  =\left\{\mn{z}: \frac{||\mn{z}||^2 -n\sigma_G^2}{n} \leq
\epsilon \right\}.
\end{equation}

With this definition, \eqref{eqn:basicAbound1} continues to hold in the
Gaussian case. 

There are two different Gaussian counterparts to
Lemma~\ref{lem:chernoffBSC}. They are both expressed in the following
lemma. 
\begin{lemma}
\label{lem:AWGNChernoff}
For Gaussian noise $Z_i$ with variance $\sigma_G^2$, 
\begin{equation} \label{eqn:mainboundraw}
\Pr(\frac{1}{n} \sum_{i=1}^n \frac{Z_i^2}{\sigma_G^2} > 1 +
\frac{\epsilon}{\sigma_G^2})
\leq \left((1 + \frac{\epsilon}{\sigma_G^2})
  \exp(-\frac{\epsilon}{\sigma_G^2})
   \right)^{\frac{n}{2}}.
\end{equation}
Furthermore 
\begin{equation} \label{eqn:mainbound}
\Pr(\frac{1}{n} \sum_{i=1}^n \frac{Z_i^2}{\sigma_G^2} > 1 +
\frac{\epsilon}{\sigma_G^2})
\leq 
\exp(-\frac{\sqrt{n} \epsilon}{4 \sigma_G^2})
\end{equation}
for all $\epsilon \geq \frac{3 \sigma_G^2}{\sqrt{n}}$.
\end{lemma}
\begin{proof}
See Appendix~\ref{app:ChernoffGaussian}.
\end{proof}
\vspace{0.1in}

To have $\Pr(\typicalc) \leq \frac{\delta}{2}$, it suffices
to pick any $\epsilon(\delta,n)$ large enough. 

So \begin{eqnarray} 
\Pr_P(A) & \geq & \int_{\mn{z}\in A \cap \typical{}}
f_P(\mn{z})d\mn{z} \nonumber\\
&=&  \int_{\mn{z}\in A \cap \typical{} }
\frac{f_P(\mn{z})}{f_G(\mn{z})}f_G(\mn{z})d\mn{z}. \label{eq:chg2}
\end{eqnarray}
Consider the ratio of the two pdf's for $\mn{z} \in \typical{}$
\begin{eqnarray}
\nonumber\frac{f_P(\mn{z})}{f_G(\mn{z})}
&=&
\left(\sqrt\frac{\sigma_G^2}{\sigma_P^2}\right)^n
\exp\left(-\|\mn{z}\|^2\left(\frac{1}{2\sigma_P^2}-\frac{1}{2\sigma_G^2}\right)\right)\\
\nonumber &\geq& 
\exp\left(-(n\sigma_G^2+n\epsilon(\delta,n))\left(\frac{1}{2\sigma_P^2}-\frac{1}{2\sigma_G^2}\right) 
 +n\ln\left(\frac{\sigma_G}{\sigma_P}\right)\right)\\
& = &\exp\left(-\frac{\epsilon(\delta,n) n}{2\sigma_G^2}\left(\frac{\sigma_G^2}{\sigma_P^2}-1\right)-nD(\sigma_G^2||\sigma_P^2)\right)
\label{eq:ratio}
\end{eqnarray}
where $D(\sigma_G^2||\sigma_P^2)$ is the KL-divergence between two
Gaussian distributions of variances $\sigma_G^2$ and $\sigma_P^2$
respectively. Substitute~\eqref{eq:ratio} back
in~\eqref{eq:chg2} to get
\begin{eqnarray}
\Pr_P(A)&\geq& \exp\left(-\frac{\epsilon(\delta,n)
  n}{2\sigma_G^2}\left(\frac{\sigma_G^2}{\sigma_P^2}-1\right)
-nD(\sigma_G^2||\sigma_P^2)\right)
\int_{\mn{z}\in A \cap \typical{}}f_G(\mn{z})d\mn{z} \nonumber \\
& \geq & \frac{\delta}{2} 
\exp\left(-nD(\sigma_G^2||\sigma_P^2) - \frac{\epsilon(\delta,n)
  n}{2\sigma_G^2}\left(\frac{\sigma_G^2}{\sigma_P^2}-1\right)\right).
\label{eqn:almostrawgaussianratiobound}
\end{eqnarray}

At this point, it is necessary to make a choice of
$\epsilon(\delta,n)$. If we are interested in studying the asymptotics
as $n$ gets large, we can use \eqref{eqn:mainbound}. This reveals that
it is sufficient to choose $\epsilon \geq \sigma_G^2
\max(\frac{3}{\sqrt{n}}, -4\frac{\ln(\delta) - \ln(2)}{\sqrt{n}})$. A
safe bet is $\epsilon = \sigma_G^2 \frac{3 +
  4\ln(\frac{2}{\delta})}{\sqrt{n}}$ or $n\epsilon(\delta,n) =
\sqrt{n}(3 + 4\ln(\frac{2}{\delta}))\sigma_G^2$. Thus
\eqref{eqn:basicAbound2} holds as well with this choice of
$\epsilon(\delta,n)$.

Substituting into \eqref{eqn:almostrawgaussianratiobound} gives 
$$\Pr_P(A) \geq \frac{\delta}{2} \exp\left(-nD(\sigma_G^2||\sigma_P^2)
-\sqrt{n}(\frac{3}{2}+2\lon{\frac{2}{\delta}})
(\frac{\sigma_G^2}{\sigma_P^2}-1) \right).$$ This establishes the
desired $f(\cdot)$ function from
\eqref{eqn:GaussianAsymptoticFdef}. To see that this function $f(x)$
is convex-$\cup$ and increasing in $x$, define $c_1 =
\exp(-nD(\sigma_G^2||\sigma_P^2) - \sqrt{n}(\frac{3}{2} + 2
\lon{2})\left(\frac{\sigma_G^2}{\sigma_P^2}-1\right) - \ln(2))$ and
$\xi = 2\sqrt{n}\left(\frac{\sigma_G^2}{\sigma_P^2}-1\right) >
0$. Then $f(\delta) = c_1 \delta \exp(\xi \ln(\delta)) = c_1
\delta^{1+\xi}$ which is clearly monotonically increasing and
convex-$\cup$ by inspection.

Attempting to use \eqref{eqn:mainboundraw} is a little more
involved. Let $\widetilde{\epsilon} = \frac{\epsilon}{\sigma_G^2}$ for
notational convenience. Then we must solve
$(1+\widetilde{\epsilon})\exp(-\widetilde{\epsilon}) =
(\frac{\delta}{2})^{\frac{2}{n}}$. Substitute
$u=1+\widetilde{\epsilon}$ to get $u\exp(-u+1) =
(\frac{\delta}{2})^{\frac{2}{n}}$. This immediately simplifies to
$-u\exp(-u) = -\exp(-1)(\frac{\delta}{2})^{\frac{2}{n}}$. At this
point, we can immediately verify that
$(\frac{\delta}{2})^{\frac{2}{n}} \in [0,1]$ and hence by the
definition of the Lambert W function in \cite{LambertWRef}, we get $u
= - W_L(-\exp(-1)(\frac{\delta}{2})^{\frac{2}{n}})$. Thus 
\begin{equation} \label{eqn:epsilonUglyDef}
\widetilde{\epsilon}(\delta,n) = -
W_L(-\exp(-1)(\frac{\delta}{2})^{\frac{2}{n}}) - 1.
\end{equation}

Substituting this into \eqref{eqn:almostrawgaussianratiobound}
immediately gives the desired expression
\eqref{eqn:GaussianNonAsymptoticFdef}. All that remains is to verify
the convexity.  Let $v = \frac{1}{2}
\left(\frac{\sigma_G^2}{\sigma_P^2}-1\right)$. As above, $f_L(\delta)
= \delta c_2 \exp(-nv \widetilde{\epsilon}(\delta,n))$. The derivatives can 
be taken using very tedious manipulations involving the relationship
$W_L'(x) = \frac{W_L(x)}{x(1 + W_L(x))}$ from \cite{LambertWRef} and
can be verified using computer-aided symbolic calculation. In our case
$-\widetilde{\epsilon}(\delta,n) = (W_L(x) + 1)$ and so this allows the
expressions to be simplified.
\begin{equation} \label{eqn:firstGaussianDeriv}
f_L'(\delta) = c_2 \exp(-nv \widetilde{\epsilon}) (2v + 1 + \frac{2v}{\widetilde{\epsilon}}). 
\end{equation}
Notice that all the terms above are positive and so the first
derivative is always positive and the function is increasing in
$\delta$. Taking another derivative gives
\begin{equation} \label{eqn:doubleGaussianDeriv}
f_L''(\delta) = c_2 
\frac{2 v(1 + \widetilde{\epsilon}) \exp(-nv \widetilde{\epsilon})}{\delta \widetilde{\epsilon}}
\left[1 + 4v + \frac{4 v}{\widetilde{\epsilon}}
             - \frac{4}{n \widetilde{\epsilon}} 
             - \frac{2}{n \widetilde{\epsilon}^2} 
             \right]. 
\end{equation}
Recall from \eqref{eqn:epsilonUglyDef} and the properties of the
Lambert $W_L$ function that $\widetilde{\epsilon}$ is a monotonically
decreasing function of $\delta$ that is $+\infty$ when $\delta = 0$
and goes down to $0$ at $\delta = 2$. Look at the term in brackets
above and multiply it by the positive $n \widetilde{\epsilon}^2$. This
gives the quadratic expression
\begin{equation} \label{eqn:keyquadratic}
(4v+1)n \widetilde{\epsilon}^2 + 4(vn - 1)\widetilde{\epsilon} - 2.
\end{equation}
This \eqref{eqn:keyquadratic} is clearly convex-$\cup$ in
$\widetilde{\epsilon}$ and negative at $\widetilde{\epsilon} = 0$.
Thus it must have a single zero-crossing for positive
$\widetilde{\epsilon}$ and be strictly increasing there. This also
means that the quadratic expression is implicitly a strictly {\em
  decreasing} function of $\delta$. It thus suffices to just check the
quadratic expression at $\delta = 1$ and make sure that it is
non-negative. Evaluating \eqref{eqn:epsilonUglyDef} at $\delta = 1$
gives $\widetilde{\epsilon}(1,n) = T(n)$ where $T(n)$ is defined in
\eqref{eqn:Tdef}.

It is also clear that \eqref{eqn:keyquadratic} is a strictly
increasing linear function of $v$ and so we can find the minimum value
for $v$ above which \eqref{eqn:keyquadratic} is guaranteed to be
non-negative. This will guarantee that the function $f_L$ is
convex-$\cup$. The condition turns out to be $v \geq \frac{2+4 T - n
  T^2}{4 n T (T+1)}$ and hence $\sigma_G^2 = \sigma_P^2 (2 v + 1) \geq
\frac{\sigma_G^2}{2} (1 + \frac{1}{T+1} + \frac{4T + 2}{nT(T+1)})$.
This matches up with \eqref{eqn:mudef} and hence the Lemma is proved.
\end{proof}

\subsection{Proof of Lemma~\ref{lem:AWGNChernoff}: Chernoff bound for
  Gaussian noise} \label{app:ChernoffGaussian}
\begin{proof}
The sum $\sum_{i=1}^n \frac{Z_i^2}{\sigma_G^2}$ is a standard $\chi^2$ random
variables with $n$ degrees of freedom. 
\begin{eqnarray}
& & \Pr(\frac{1}{n} \sum_{i=1}^n \frac{Z_i^2}{\sigma_G^2} > 1 +
  \frac{\epsilon}{\sigma_G^2}) \nonumber \\ 
& \leq_{(a)} & \inf_{s > 0} \left(\frac{\exp(-s(1 + \frac{\epsilon}{\sigma_G^2}))}{\sqrt{1 - 2s}}\right)^n \nonumber \\
& \leq_{(b)} & \left(\sqrt{1 + \frac{\epsilon}{\sigma_G^2}}
  \exp(-\frac{\epsilon}{2\sigma_G^2})
   \right)^n 
\label{eqn:basicChiChernoffBound} \\
& = & \left((1 + \frac{\epsilon}{\sigma_G^2})
  \exp(-\frac{\epsilon}{\sigma_G^2})
   \right)^{\frac{n}{2}}
\end{eqnarray}
where (a) follows using standard moment generating functions for
$\chi^2$ random variables and Chernoff bounding arguments and (b)
results from the substitution $s = \frac{\epsilon}{2(\sigma_G^2 +
  \epsilon)}$. This establishes \eqref{eqn:mainboundraw}.

For tractability, the goal is to replace
\eqref{eqn:basicChiChernoffBound} with a exponential of an affine
function of $\frac{\epsilon}{\sigma_G^2}$. For notational convenience,
let $\widetilde{\epsilon} = \frac{\epsilon}{\sigma_G^2}$. The idea is to bound
the polynomial term $\sqrt{1+\widetilde{\epsilon}}$ with an exponential as long
as $\widetilde{\epsilon} > \epsilon^*$.

Let $\epsilon^* = \frac{3}{\sqrt{n}}$ and let $K = \frac{1}{2} -
  \frac{1}{4 \sqrt{n}}$. Then it is clear that
\begin{equation} \label{eqn:polydominated}
 \sqrt{1 + \widetilde{\epsilon}} \leq \exp(K \widetilde{\epsilon})
\end{equation}
as long as $\widetilde{\epsilon} \geq \epsilon^*$. First, notice that the two
agree at $\widetilde{\epsilon} = 0$ and that the slope of the concave-$\cap$
function $\sqrt{1 + \widetilde{\epsilon}}$ there is $\frac{1}{2}$. Meanwhile, the
slope of the convex-$\cup$ function $\exp(K \widetilde{\epsilon})$ at $0$ is $K <
\frac{1}{2}$. This means that $\exp(K \widetilde{\epsilon})$ starts out below
$\sqrt{1 + \widetilde{\epsilon}}$. However, it has crossed to the other side by
$\widetilde{\epsilon} = \epsilon^*$. This can be verified by taking the logs of
both sides of \eqref{eqn:polydominated} and multiplying them both by
$2$. Consider the LHS evaluated at $\epsilon^*$ and lower-bound it by
a third-order power-series expansion
\begin{eqnarray*}
\ln(1 + \frac{3}{\sqrt{n}})
& \leq &
\frac{3}{\sqrt{n}}
- \frac{9}{2 n}
+ \frac{9}{n^{3/2}}.
\end{eqnarray*}
meanwhile the RHS of \eqref{eqn:polydominated} can be dealt with exactly:
\begin{eqnarray*}
2 K \epsilon^*
& = &
(1 - \frac{1}{2\sqrt{n}})\frac{3}{\sqrt{n}} \\
& = &
\frac{3}{\sqrt{n}}
- \frac{3}{2 n}.
\end{eqnarray*}

For $n \geq 9$, the above immediately establishes
\eqref{eqn:polydominated} since $\frac{9}{2 n} - \frac{3}{2 n} =
\frac{3}{n} \geq \frac{9}{n\sqrt{9}}$. The cases $n=1,2,3,4,5,6,7,8$
can be verified by direct computation. Using
\eqref{eqn:polydominated}, for $\widetilde{\epsilon} > \epsilon^*$ we have:
\begin{eqnarray}
\Pr(\typicalc) & \leq &
[\exp(K \widetilde{\epsilon})\exp(-\frac{1}{2} \widetilde{\epsilon})]^n \nonumber \\
& = & \exp(-\frac{\sqrt{n}}{4} \widetilde{\epsilon}). \label{eqn:combinedaffinishbound}
\end{eqnarray}
\end{proof}

\section{Approximation analysis for the BSC} \label{app:Taylorexpansions}

\subsection{Lemma \ref{lem:hbbound}} \label{app:lemhbbound}

\begin{proof}
\eqref{eqn:binentropylowerbound} and
\eqref{eqn:invbinentropyupperbound} are obvious from the
concave-$\cap$ nature of the binary entropy function and its values at
$0$ and $\frac{1}{2}$.

\begin{eqnarray*}
h_b(x)& = &x \lo{1/x} + (1-x) \lo{1/(1-x)}\\
&\leq_{(a)} & 2 x \lo{1/x} =  2 x \ln(1/x)/\ln(2)\\
&\leq_{(b)} & 2 x d(\frac{1}{x^{1/d}} - 1)/\ln(2)\;\;\;\;\forall d>1\\
& \leq & 2 x^{1-1/d} d/\ln(2).
\end{eqnarray*}
Inequality $(a)$ follows from the fact that $x^x<(1-x)^{1-x}$ for $x\in (0,\frac{1}{2})$.
For inequality $(b)$, observe that $\ln(x)\leq x-1$. This implies
$\ln(x^{1/d})\leq x^{1/d} -1$. Therefore, $\ln(x)\leq d(x^{1/d}-1)$ for
all $x>0$ since $\frac{1}{d} \leq 1$ for $d \geq 1$.

The bound on $h_b^{-1}(x)$ follows immediately by identical arguments.
\end{proof}

\subsection{Lemma \ref{lem:approxloghbinvdeltag}} \label{app:lemapproxloghbinvdeltag}
\begin{proof}

First, we investigate the small gap asymptotics for $\delta(g^*)$, where
$g^*=p+gap^r$ and $r < 1$.
\begin{eqnarray}
\delta(g^*)& = &1- \frac{C(g^*)}{R} \nonumber \\
&=& 1 -\frac{C(p+gap^r)}{C(p)-gap} \nonumber \\
&=& 1 -\frac{C(p)-gap^rh_b'(p)+o(gap^r)}{C(p)(1-gap/C(p))} \nonumber \\
&=& 1 - (1 - \frac{h_b'(p)}{C(p)}gap^r +o(gap^r)) \times (1+gap/C(p) + o(gap)) \nonumber \\ 
&=& \frac{h_b'(p)}{C(p)}gap^r  + o(gap^r). \label{eqn:deltaExpansion}
\end{eqnarray}

Plugging \eqref{eqn:deltaExpansion} into
\eqref{eqn:invbinentropyupperbound} and using Lemma~\ref{lem:hbbound}
gives 
\begin{eqnarray}
\lo{h_b^{-1}\left(\delta(g^*)\right)}
&\leq&\lo{\frac{h_b'(p)}{2 C(p)}gap^r  + o(gap^r)} \\
&=&\lo{\frac{h_b'(p)}{2 C(p)}} + r\lo{gap} + o(1) \\
&=& r\lo{gap} -1 + \lo{\frac{h_b'(p)}{C(p)}}  + o(1)
\end{eqnarray}
and this establishes the upper half of \eqref{eqn:approxloghbinvdeltag}. 

To see the lower half, we use \eqref{eqn:invbinentropylowerbound}:
\begin{eqnarray*}
\lo{h_b^{-1}\left(\delta(g^*)\right)} 
&\geq& 
\frac{d}{d-1}\left(\lo{\delta(g^*)}+\lo{\frac{\ln{2}}{2d}} \right) \\
&=&
\frac{d}{d-1}\left(
\log_2\left(\frac{h_b'(p)}{C(p)}gap^r +o(gap^r)\right)
+\lo{\frac{\ln{2}}{2d}} \right) \\
&=&
\frac{d}{d-1}\left(
r\lo{gap}+\lo{\frac{h_b'(p)}{C(p)}} + o(1) 
+\lo{\frac{\ln{2}}{2d}} \right) \\
&=& \frac{d}{d-1}r\;\lo{gap} - 1 + K_1  + o(1)
\end{eqnarray*}
where $K_1 = \frac{d}{d-1}\left(\lo{\frac{h_b'(p)}{C(p)}} +
\lo{\frac{\ln(2)}{d}}\right)$ and $d > 1$ is arbitrary.
\end{proof}

\subsection{Lemma \ref{lem:approxdiv}} \label{app:lemapproxdiv}
\begin{proof} 
\begin{eqnarray*}
D(g^*||p) &=& D(p+gap^r||p)\\
&=& 0 + 0 \times gap^r + \frac{1}{2} \frac{gap^{2r}}{p(1-p)\ln(2)} +o(gap^{2r})
\end{eqnarray*}
since $D(p||p) = 0$ and the first derivative is also zero. Simple
calculus shows that the second derivative of $D(p+x||p)$ with respect
to $x$ is $\frac{\lo{e}}{(p+x)(1-p-x)}$. 
\end{proof}

\subsection{Lemma \ref{lem:approxlogfrac}} \label{app:lemapproxlogfrac}
\begin{proof}
\begin{eqnarray*}
\log_2\left(\frac{g^*(1-p)}{p(1-g^*)}\right)
& = & \lo{\frac{1-p}{p}}+\lo{\frac{g^*}{1-g^*}}\\
& = & \lo{\frac{1-p}{p}}+\lo{g^*} - \lo{1-g^*}\\
& = & \lo{\frac{1-p}{p}}+\lo{p+gap^r} - \lo{1-p-gap^r}\\
& = & \lo{\frac{1-p}{p}}+\lo{p} + \lo{1+\frac{gap^r}{p}}
      -\lo{1-p} - \lo{1-\frac{gap^r}{1-p}}\\
& = & \frac{gap^r}{p\ln(2)} + \frac{gap^{r}}{(1-p)\ln(2)} + o(gap^r)\\
& = & \frac{gap^r}{p(1-p)\ln(2)}+o(gap^r) \\
& = & \frac{gap^r}{p(1-p)\ln(2)}(1 +o(1)).
\end{eqnarray*}
\end{proof}

\subsection{Lemma \ref{lem:approxepsilon}} \label{app:lemapproxepsilon}
\begin{proof}
Expand \eqref{eqn:epsilondef}:
\begin{eqnarray*}
\epsilon 
&= & \sqrt{\frac{1}{K(p + gap^r)}}
 \sqrt{\lo{\frac{2}{h_b^{-1}(\delta(G))}}} \\
&= & \sqrt{\frac{1}{\ln(2) K(p + gap^r)}}
 \sqrt{\ln(2) - \ln(h_b^{-1}(\delta(G)))} \\
&\geq & \sqrt{\frac{1}{\ln(2) K(p + gap^r)}}
 \sqrt{\ln(2) - r\ln(gap) +\ln(2) - K_2\ln(2) + o(1)} \\
& = & \sqrt{\frac{1}{\ln(2) K(p + gap^r)}}
 \sqrt{r\ln(\frac{1}{gap}) + (2 - K_2)\ln(2) + o(1)} \\
& = & \sqrt{\frac{1}{\ln(2) K(p + gap^r)}}
 \sqrt{r\ln(\frac{1}{gap})} (1 + o(1)).
\end{eqnarray*} 
and similarly
\begin{eqnarray*}
\epsilon 
& = & \sqrt{\frac{1}{\ln(2) K(p + gap^r)}}
 \sqrt{\ln(2) - \ln(h_b^{-1}(\delta(G)))} \\
& \leq & \sqrt{\frac{1}{\ln(2) K(p + gap^r)}}
\sqrt{(2 - K_2)\ln(2) + \frac{d}{d-1}r\lon{\frac{1}{gap}} + o(1)}  \\
&=& \sqrt{\frac{r d}{\ln(2) (d-1)K(p + gap^r)}} \sqrt{\lon{\frac{1}{gap}}} 
 (1 + o(1)).
\end{eqnarray*}  
All that remains is to show that $K(p+gap^r)$ converges to $K(p)$ as
$gap\rightarrow 0$. Examine \eqref{eqn:Kdef}. The continuity of
$\frac{D(g+\eta ||g)}{\eta^2}$ is clear in the interior $\eta \in
(0,1-g)$ and for $g \in (0,\frac{1}{2})$. All that remains is to check
the two boundaries. $\lim_{\eta \rightarrow 0} \frac{D(g+\eta
  ||g)}{\eta^2} = \frac{1}{g(1-g)\ln 2}$ by the Taylor expansion of
$D(g+\eta ||g)$ as done in the proof of
Lemma~\ref{lem:approxdiv}. Similarly, $\lim_{\eta \rightarrow 1-g} \frac{D(g+\eta
  ||g)}{\eta^2} = D(1||g) = \lo{\frac{1}{1-g}}$. Since $K$ is a
minimization of a continuous function over a compact set, it is itself
continuous and thus the limit $\lim_{gap \rightarrow 0} K(p+gap^r) =
K(p)$.

Converting from natural logarithms to base 2 completes the proof.
\end{proof}

\subsection{Approximating the solution to the quadratic formula}
\label{app:nbsc}
In \eqref{eq:quad}, for $g=g^*=p+gap^r$,
\begin{eqnarray*}
a &=& D(g^*||p)\\
b &=& \epsilon \log_2\left( \frac{g^*(1-p)}{p(1-g^*)}   \right)\\
c &=& \lo{\pe_P} - \lo{h_b^{-1}(\delta(g^*))} +1.
\end{eqnarray*}
The first term, $a$, is approximated by Lemma~\ref{lem:approxdiv} so
\begin{equation} \label{eqn:aApprox}
a = gap^{2r}(\frac{1}{2 p(1-p)\ln(2)}+o(1)).
\end{equation}
Applying Lemma~\ref{lem:approxlogfrac} and
Lemma~\ref{lem:approxepsilon} reveals 
\begin{eqnarray} 
b & \leq & \sqrt{\frac{r d }{(d-1)K(p)}} \sqrt{\lo{\frac{1}{gap}}}
\frac{gap^r}{p(1-p)\ln(2)}(1+o(1)) \nonumber \\
& = & 
\frac{1}{p(1-p) \ln(2)} \sqrt{\frac{r d}{ (d-1)K(p)}}
\sqrt{gap^{2r} \lo{\frac{1}{gap}}}(1 + o(1))
 \label{eqn:bApproxUpper} \\
b & \geq & 
\frac{1}{p(1-p)  \ln(2)} \sqrt{\frac{r}{K(p)}}
\sqrt{gap^{2r} \lo{\frac{1}{gap}}}(1 + o(1)).
 \label{eqn:bApproxLower}
\end{eqnarray}
The third term, $c$, can be bounded similarly using
Lemma~\ref{lem:approxloghbinvdeltag} as follows,  
\begin{eqnarray}
c 
&=&   \beta \lo{gap} - \lo{h_b^{-1}(\delta(g^*))} +1 \nonumber \\
&\leq& (\frac{d}{d-1} r - \beta) \lo{\frac{1}{gap}} +
K_3 + o(1) \label{eqn:cApproxUpper} \\
&\geq& (r - \beta) \lo{\frac{1}{gap}} +
K_4 + o(1). \label{eqn:cApproxLower}
\end{eqnarray}
for a pair of constants $K_3,K_4$. Thus, for $gap$ small enough and
$r < \frac{\beta (d-1)}{d}$, we know that $c<0$.  

The lower bound on $\sqrt{n}$ is thus
\begin{eqnarray}
\nonumber\sqrt{n}&\geq & \frac{\sqrt{b^2-4ac}-b}{2a}\\
&=&\frac{b}{2a}\left(\sqrt{1-\frac{4ac}{b^2}}-1\right).
\label{eq:sqrtnbd}
\end{eqnarray}

Plugging in the bounds \eqref{eqn:aApprox} and
\eqref{eqn:bApproxLower} reveals that
\begin{equation}
\label{eq:bby2a}
\frac{b}{2a} 
\geq 
\frac{\sqrt{\lo{\frac{1}{gap}}}}{gap^r}
\sqrt{\frac{r}{K(p)}}
(1 + o(1))
\end{equation}
Similarly, using \eqref{eqn:aApprox}, \eqref{eqn:bApproxLower},
\eqref{eqn:cApproxUpper}, we get
\begin{eqnarray}
\frac{4ac}{b^2}
&\leq &
\frac{4gap^{2r} \left(\frac{1}{p(1-p)\ln(2)}\right) \times
  \left[ (\frac{d}{d-1}r - \beta) \lo{\frac{1}{gap}} + K_3\right](1+o(1))
}{(\frac{1}{p(1-p) \ln(2)})^2 \frac{r}{K(p)} gap^{2r}
  \lo{\frac{1}{gap}} (1 + o(1)) } \nonumber \\
&=& 4 p(1-p) K(p) \ln(2) \left[ \frac{d}{d-1} - \frac{\beta}{r}\right]  +
o(1). \label{eq:bby2b}
\end{eqnarray}
This tends to a negative constant since $r < \frac{\beta (d-1)}{d}$. 

Plugging \eqref{eq:bby2a} and \eqref{eq:bby2b} into
\eqref{eq:sqrtnbd} gives:

\begin{eqnarray}
n
&\geq&
[\sqrt{\frac{r}{K(p)}}
\frac{\lo{\frac{1}{gap}}}{gap^r} (1 + o(1))
\left(\sqrt{1 + 4 p(1-p) \ln(2) K(p) \left[ \frac{\beta}{r} - \frac{d}{d-1}
    \right] +o(1)}
-1\right)]^2 \nonumber \\
& = & [\frac{\lo{\frac{1}{gap}}}{gap^r}]^2
\frac{1}{K(p)}
\left(\sqrt{r + 4 p(1-p) \ln(2) K(p) \left[ \beta - \frac{r d}{d-1} \right]} - \sqrt{r}
\right)^2 (1 + o(1)) \nonumber \\
& = & \Omega \left( \frac{(\lo{1/gap})^2}{gap^{2r}}  \right)
\end{eqnarray}
for all $r\leq \min\{\frac{\beta}{\frac{d}{d-1}},1\}$. By taking $d$ arbitrarily large, we arrive at Theorem~\ref{thm:lbn} for the BSC.

\section{Approximation analysis for the AWGN channel}
\label{app:approxawgn}
Taking logs on both sides of~\eqref{eq:lbnawgn} for a fixed test channel $G$,
\begin{equation}
\ln(\pe_P) \geq \ln(h_b^{-1}(\delta(G))) - \ln(2) -
nD(\sigma_G^2||\sigma_P^2) - \sqrt{n}(\frac{3}{2} + 2\ln 2 - 2
\ln(h_b^{-1}(\delta(G)))\left(\frac{\sigma_G^2}{\sigma_P^2}-1\right), 
\end{equation}
Rewriting this in the standard quadratic form using 
\begin{eqnarray}
a &=& D(\sigma_{G^*}^2||\sigma_P^2), \label{eq:aGaussDef}\\
b &=& (\frac{3}{2} + 2\ln 2 - 2
\ln(h_b^{-1}(\delta(G))))\left(\frac{\sigma_G^2}{\sigma_P^2}-1\right),
\label{eq:bGaussDef} \\
c &= & \ln(\pe_P) - \ln(h_b^{-1}(\delta(G))) +\ln(2) \label{eq:cGaussDef}.
\end{eqnarray}
it suffices to show that the terms exhibit behavior as $gap\rightarrow
0$ similar to their BSC counterparts. 
 
For Taylor approximations, we use the channel $G^*$, with
corresponding noise variance $\sigma_{G^*}^2 = \sigma_P^2 + \zeta$,
where  
\begin{equation}
\label{eq:zeta}
\zeta = gap^r \left(\frac{2 \sigma_P^2 (P_T+\sigma_P^2)}{P_T}\right).
\end{equation}
\begin{lemma}
For small enough $gap$, for $\zeta$ as in~\eqref{eq:zeta}, if $r<1$ then $C({G^*})<R$.
\end{lemma}
\begin{proof}
Since $C(P)-gap=R>C({G^*})$, we must satisfy
\begin{eqnarray}
gap &\leq&
\frac{1}{2}\lo{1+\frac{P_T}{\sigma_P^2}}-\frac{1}{2}\lo{1+\frac{P_T}{\sigma_P^2+\zeta}}.
\nonumber 
\end{eqnarray}
So the goal is to lower bound the RHS above to show that
\eqref{eq:zeta} is good enough to guarantee that this is bigger than
the gap. 
So
\begin{eqnarray}
& = & \frac{1}{2}\left(\lo{1 + \frac{\zeta}{\sigma_P^2}} - \lo{1 +
  \frac{\zeta}{\sigma_P^2 + P_T}} \right) \nonumber \\
& = & \frac{1}{2}
\left(\lo{1 + 2 gap^r(1 + \frac{\sigma_P^2}{P_T})}
    - \lo{1 + 2 gap^r \frac{\sigma_P^2}{P_T}} \right) \nonumber \\
& \geq & \frac{1}{2}\left(\frac{c_s}{\ln(2)} 2gap^r(1 + \frac{\sigma_P^2}{P_T}) 
    - \frac{1}{\ln(2)} 2 gap^r \frac{\sigma_P^2}{P_T} \right) \nonumber \\
& = & gap^r \frac{1}{\ln(2)} \left(c_s - ( 1-c_s)\frac{\sigma_P^2}{P_T} \right).
\label{eq:satisfy}
\end{eqnarray}
For small enough $gap$, this is a valid lower bound as long as $c_s <
1$. Choose $c_s$ so that $1 < c_s <
\frac{\sigma_P^2}{P_T + \sigma_P^2}$. For $\zeta$ as
in~\eqref{eq:zeta}, the LHS is $gap^r K$ and thus clearly having $r<1$
suffices for satisfying~\eqref{eq:satisfy} for small enough
$gap$. This is because the derivative of $gap^r$ tends to infinity as
$gap \rightarrow 0$. 
\end{proof}
\vspace{0.1in}

In the next Lemma, we perform the approximation analysis for the terms
inside~\eqref{eq:aGaussDef}, \eqref{eq:bGaussDef} and
\eqref{eq:cGaussDef}.

\begin{lemma}  \label{lem:gaussianApproximations}
Assume that $\sigma_{G^*}^2 = \sigma_P^2 + \zeta$ where $\zeta$ is
defined in \eqref{eq:zeta}.
\begin{itemize}
\item[(a)] \begin{equation}
  \frac{\sigma_{G^*}^2}{\sigma_P^2}-1 = 
  gap^r \left(\frac{2 (P_T + \sigma_P^2)}{P_T} \right). \end{equation} 
\item[(b)] \begin{equation}
 \ln(\delta({G^*})) = r\ln(gap) + o(1)-\ln(C(P)). \end{equation}
\item[(c)] \begin{equation}
\ln(h_b^{-1}(\delta({G^*}))) \geq
  \frac{d}{d-1}r\ln(gap) + c_2,\end{equation} 
  for some constant $c_2$ that is a function of $d$.
\begin{equation}\ln(h_b^{-1}(\delta({G^*}))) \leq
  r\ln(gap) + c_3,\end{equation} for some constant $c_3$.
\item[(d)] \begin{equation}
D(\sigma_{G^*}^2||\sigma_P^2) =
 \frac{(P_T + \sigma_P^2)^2}{P_T^2} gap^{2r}(1 + o(1)).\end{equation} 
\end{itemize}
\end{lemma}
\begin{proof}
\textit{(a) } Immediately follows from the definitions and \eqref{eq:zeta}.

\textit{(b) }We start with simplifying $\delta({G^*})$
\begin{eqnarray*}
\delta({G^*})&=&1-\frac{C({G^*})}{R}\\
&=&\frac{C-gap-\frac{1}{2} \lo{1+\frac{P_T}{\sigma_{G^*}^2}}}{C-gap}\\
&=&\frac{\frac{1}{2}\lo{1+\frac{P_T}{\sigma_P^2}}-\frac{1}{2}\lo{1+\frac{P_T}{\sigma_P^2+\zeta}}-gap}{C-gap}\\
&=&\frac{\frac{1}{2}\lo{(\frac{\sigma_P^2 +
      P_T}{\sigma_P^2})(\frac{\sigma_P^2+\zeta}{P_T + \sigma_P^2+\zeta})}-gap}{C-gap}\\
&=&\frac{\frac{1}{2}\lo{1+\frac{\zeta}{\sigma_P^2}}-\frac{1}{2}\lo{1+\frac{\zeta}{P_T
      +\sigma_P^2}}-gap}{C-gap} \\
&=&\frac{\frac{1}{2}\frac{\zeta}{\sigma_P^2}-\frac{1}{2}\frac{\zeta}{P_T
    +\sigma_P^2}+o(\zeta)-gap}{C-gap}\\
&=& \frac{\frac{1}{2}\left( \frac{\zeta P_T}{\sigma_P^2(P_T+\sigma_P^2)} +o(\zeta)\right)-gap}{C-gap}\\
&=& \frac{1}{C}(\frac{1}{2}\left( gap^r \frac{2 \sigma_P^2  (P_T+\sigma_P^2)}{P_T} \frac{P_T}{\sigma_P^2(P+\sigma_P^2)}
+o(gap^r)\right)-gap)(1 - \frac{gap}{C} + o(gap)) \\
&=& \frac{gap^r}{C}(1 +o(1)).
\end{eqnarray*}
Taking $\ln(\cdot{})$ on both sides, the result is evident.

\textit{(c)} follows from \textit{(b)} and Lemma~\ref{lem:hbbound}.

\textit{(d)} comes from the definition of
$D(\sigma_{G^*}^2||\sigma_P^2)$ followed immediately by the expansion 
$\ln(\sigma_{G^*}^2/\sigma_P^2)=\ln(1+\zeta/\sigma_P^2) = 
\frac{\zeta}{\sigma_P^2} - \frac{1}{2} (\frac{\zeta}{\sigma_P^2})^2 +
o(gap^{2r})$. All the constant and first-order in $gap^r$ terms
cancel since $\frac{\sigma_{G^*}^2}{\sigma_P^2} = 1 +
\frac{\zeta}{\sigma_P^2}$. This gives the result immediately. 
\end{proof}

\vspace{0.1in}

Now, we can use Lemma~\ref{lem:gaussianApproximations} to approximate
\eqref{eq:aGaussDef}, \eqref{eq:bGaussDef} and \eqref{eq:cGaussDef}.

\begin{eqnarray}
a & = & \frac{(P_T + \sigma_P^2)^2}{P_T^2} gap^{2r}(1 + o(1)) \label{eqn:aGaussApprox} \\
b & = &
\left(\frac{3}{2} + 2\ln 2 - 2\ln(h_b^{-1}\left(\delta(G)\right)) \right)
gap^r \frac{2 (P_T + \sigma_P^2)}{P_T} \nonumber \\
& \leq & 
\frac{2 d (P_T+\sigma_P^2)}{(d-1)P_T}  r\ln(\frac{1}{gap}) gap^r (1 + o(1)) \label{eqn:bGaussUpper} \\
b &\geq & 
\frac{2(P_T+\sigma_P^2)}{P_T} r\ln(\frac{1}{gap}) gap^r (1 + o(1)) \label{eqn:bGaussLower} \\
c & \leq &(\frac{d}{d-1}r -\beta) \ln(\frac{1}{gap})(1 + o(1)) \label{eqn:cGaussUpper}\\
c & \geq &(r -\beta) \ln(\frac{1}{gap})(1 + o(1)). \label{eqn:cGaussLower}
\end{eqnarray}

Therefore, in parallel to \eqref{eq:bby2a}, we have for the AWGN bound
\begin{equation} \label{eq:bbG2a}
\frac{b}{2a} \geq \frac{r P_T}{(P_T + \sigma_P^2)} \left(\frac{\ln(\frac{1}{gap})}{gap^r}\right)(1
+ o(1)).
\end{equation}
Similarly, in parallel to \eqref{eq:bby2b}, we have for the AWGN bound
\begin{eqnarray*}
\frac{4ac}{b^2} &\leq& 
(1 + o(1)) 
\frac{1}{r^2}(\frac{d}{d-1}r -\beta) \frac{1}{\ln(\frac{1}{gap})}.
\end{eqnarray*}
This is negative as long as $r < \frac{\beta (d-1)}{d}$, and so for
every $c_S < \frac{1}{2}$ for small enough $gap$, we know that 
\begin{eqnarray*}
\sqrt{1-\frac{4ac}{b^2}} -1&\geq& c_s \frac{1}{r^2}(\beta -
\frac{d}{d-1}r) \frac{1}{\ln(\frac{1}{gap})} (1 + o(1)).
\end{eqnarray*} 
Combining this with \eqref{eq:bbG2a} gives the bound:
\begin{eqnarray}
n & \geq & (1+o(1))[c_s \frac{1}{r^2}(\beta -
\frac{d}{d-1}r) \frac{1}{\ln(\frac{1}{gap})}  
\frac{r P_T}{P_T + \sigma_P^2}
\left(\frac{\ln(\frac{1}{gap})}{gap^r}\right)]^2 \\
& = & (1+o(1))[c_s \frac{P_T}{r (P_T + \sigma_P^2)}
(\beta - \frac{d}{d-1}r)
\left(\frac{1}{gap^r}\right)]^2.
\end{eqnarray}
Since this holds for all $0 < c_s < \frac{1}{2}$ and all $r <
\min(1,\frac{\beta (d-1)}{d})$ for all $d > 1$, Theorem~\ref{thm:lbn}
for AWGN channels follows. 

\section*{Acknowledgments}
Years of conversations with colleagues in the Berkeley Wireless
Research Center have helped motivate this investigation and informed
the perspective here. Cheng Chang was involved with the discussions
related to this paper, especially as regards the AWGN case. Sae-Young
Chung (KAIST) gave valuable feedback at an early stage of this
research and Hari Palaiyanur caught many typos in early drafts of this
manuscript. Funding support from NSF CCF 0729122, NSF ITR 0326503, NSF
CNS 0403427, and gifts from Sumitomo Electric.

\bibliographystyle{IEEEtran}
\bibliography{IEEEabrv,MyMainBibliography}

\end{document}